\documentclass[runningheads]{llncs}

\usepackage{graphicx}
\usepackage{booktabs}   %% For formal tables:
\usepackage{subcaption} %% For complex figures with subfigures/subcaptions
\usepackage{mathtools}
\usepackage{dsfont}
\usepackage{paralist}
\usepackage{amscd}
\usepackage{tikz}
\usetikzlibrary{matrix}
\usetikzlibrary{fit}
\usetikzlibrary{automata}
\usetikzlibrary{arrows,shapes,backgrounds,positioning}
\usepackage{multirow}
\usepackage{dashbox}
\usepackage{amsfonts}
\usepackage[normalem]{ulem}
\usepackage{wrapfig}
\usepackage{appendix}
\usepackage{algpseudocode}
\usepackage[linesnumbered,ruled,noend]{algorithm2e}
\algnewcommand\algorithmicforeach{\textbf{for each}}
\algdef{S}[FOR]{ForEach}[1]{\algorithmicforeach\ #1\ \algorithmicdo}
\usepackage{listings}
\usepackage{amssymb}
\usepackage{amsmath,calc}
\usepackage{array}                                              % http://ctan.org/pkg/array

\usepackage{floatrow}
\newfloatcommand{capbtabbox}{table}[][\FBwidth]
\usepackage{blindtext}

\newcolumntype{C}[1]{>{\centering\arraybackslash$}m{#1}<{$}}
\newlength{\mycw}                                         % array column width
\usepackage[utf8]{inputenc}

\usepackage[T1]{fontenc}
\usepackage{textcase}
\usepackage[debug=true]{hyperref}

\usepackage{xspace}

\newcommand{\Nats}{\ensuremath{\mathbb{N}}\xspace}
\newcommand{\trans}[3]{\ensuremath{#1 \xrightarrow{#2} #3}\xspace}

\newcommand{\conf}{\mathbf{q}}
\newcommand{\altconf}{\mathbf{p}}
\newcommand{\wqo}{\preceq}
\newcommand{\altwqo}{\trianglelefteq}
\newcommand{\nice}{synchronization-independent\xspace}

\newcommand{\ie}{i.e.\xspace}

\newcommand{\GS}{Q}
\newcommand{\LS}{S}
\newcommand{\ls}{s}
\newcommand{\guards}{\mathcal{P}(\LS)}
\newcommand{\guard}{G}
\newcommand{\guardset}{\mathcal{\guard}}
\newcommand{\free}{free\xspace}

\newcommand{\senderset}{\hat{\ls}\xspace}

\newcommand{\ra}{$\rightarrow$} 
\newcommand{\st}[1]{\textsc{#1}} 
\newcommand{\tr}[1]{$\mathbf{#1}$} 
\newcommand{\trr}[1]{\mathbf{#1}}

\newcommand{\lemref}[1]{Lem.~\ref{lem:#1}}

\newcommand{\secref}[1]{Sec.~\ref{sec:#1}}

\newcommand{\figref}[1]{Fig.~\ref{fig:#1}}
\newcommand{\exref}[1]{Ex.~\ref{ex:#1}}

\newcommand{\thmref}[1]{Thm.~\ref{thm:#1}}

\newcommand{\appref}[1]{App.~\hyperref[app:#1]{\ref*{app:#1}}}

\newcommand{\tabref}[1]{Tab.~\ref{table:#1}}

\newcommand{\spara}[1]{\medskip \noindent {\bf #1}}

\newcommand{\mtt}[1]{\mathtt{#1}}
\newcommand{\Nat}{\mathbb{N}}

\newcommand\tuple[1]{\langle #1 \rangle}

\newcommand{\consagree}{agreement\xspace}

\newcommand{\GActions}{\mathcal{A}}

\newcommand{\act}{a}
\newcommand{\actname}{\mtt{msg}}

\newcommand{\M}{\mathcal{M}}

\newcommand{\ch}{\textsc{Choose}\xspace} 
 
\newcommand{\chfig}{\code{cons}\xspace} 
\newcommand{\chset}{\mathcal{S}}	%input set to choose
\newcommand{\supp}{\mathsf{supp}\xspace}

\newcommand{\gbc}{GSP\xspace} %at least for now
\newcommand{\hlm}{\gbc model\xspace}	%Swen's model (high-level model)
\newcommand{\ph}{P_{\gbc}} %local process definition in the gbc model

\newcommand{\amenable}{\textbf{cutoff-amenable}\xspace}

\newcommand{\wellbehaved}{\textbf{well-behaved}\xspace} 

\definecolor{Gray}{gray}{0.85}
\definecolor{LightRed}{rgb}{1,0.7,0.7}
\definecolor{LightGreen}{rgb}{0.7,1,0.7}

\newcommand{\redbox}[1]{\color {red} \setlength{\fboxrule}{1.5pt} \dbox {\color {black}#1}\color {black} \ }

\definecolor{realGreen}{rgb}{0.0, 0.5, 0.0}
\newcommand{\originalgrumblere}[2]{\textcolor{realGreen}{\sl{\bf #1:\\} #2}}

\newcommand{\originalgrumblern}[2]{\textcolor{red}{\sl{\bf #1:} #2}}

\newcommand{\originalgrumblerr}[2]{\textcolor{magenta}{\sl{\bf #1 says:} #2}}

\newcommand{\grumblere}[2]{\originalgrumblere{#1}{#2}}

\newcommand{\grumblern}[2]{\originalgrumblern{#1}{#2}}

\newcommand{\grumblerr}[2]{\originalgrumblerr{#1}{#2}}

\newcommand{\nour}[1]{\grumblern{Nour}{#1}}

 \newcommand{\roopsha}[1]{\grumblerr{Roopsha}{#1}}
 
 \newcommand{\edit}[1]{\grumblere{\\ Edit}{#1} \\}

\newcommand{\mnx}[1]{}

\newcommand{\Env}{Env\xspace}
\newcommand{\Ask}{Ask\xspace}
\newcommand{\Pick}{Pick\xspace}
\newcommand{\Report}{Report\xspace}
\newcommand{\Idle}{Idle\xspace}

\newcommand{\Choose}{Choose\xspace}
\newcommand{\Smoke}{Smoke\xspace}
\newcommand{\Resett}{Reset\xspace}

\definecolor{bluekeywords}{rgb}{0.13,0.13,1}
\definecolor{greencomments}{rgb}{0,0.5,0}
\definecolor{turqusnumbers}{rgb}{0.17,0.57,0.69}
\definecolor{redstrings}{rgb}{0.5,0,0}

\lstdefinelanguage{FSharp}
    {keywords=[1]{if,else,machine, on, unit, env, contains, goto, Consensus, state, variables, actions, initial,do,Send,Broadcast,bool,Set}, 
    keywordstyle=[1]\color{bluekeywords},
    keywords=[2]{otherstuff},  % this can be used to have other words show in a different style.
    keywordstyle=[2]\color{redstrings},
    sensitive=false,
    morecomment=[l][\color{greencomments}]{///},
    morecomment=[l][\color{greencomments}]{//},
    morecomment=[s][\color{greencomments}]{{(*}{*)}},
    morestring=[b]",
    stringstyle=\color{redstrings}
    }

    \lstnewenvironment{dsl}
  {
    \lstset{
        language=FSharp,
        basicstyle=\ttfamily,
        breaklines=true,
        columns=fullflexible}
  }
  {
  }

\newcommand{\code}[1]{\texttt{#1}} %at least for now
\newcommand{\reducespace}{-1pt}
\newcommand{\squeezecaption}{\vspace{\reducespace}}

\renewcommand{\paragraph}{\spara}

\begin{document}
\title{Parameterized Verification of Systems\\ with Global Synchronization and Guards\thanks{This research was partially supported by the National Science
    Foundation under Grant
  Nos.~1846327,~1908504, and~1919197 and by a grant from the Purdue Research
  Foundation. 
  Any opinions, findings, and conclusions  in  this  paper are
  those  of  the  authors  only and  do  not necessarily reflect the views of
  our sponsors.}}

\titlerunning{Parameterized Verification of Systems with Global Synchr. \& Guards}
\author{Nouraldin Jaber\inst{1}\thanks{Joint first-authors.} \and
Swen Jacobs\inst{2}$^{\star\star}$ \and
Christopher Wagner\inst{1} \and
Milind Kulkarni\inst{1} \and
Roopsha Samanta\inst{1}
}
\authorrunning{N. Jaber, S. Jacobs, C. Wagner, M. Kulkarni, and R. Samanta}

%Second Author\inst{2,3}\orcidID{1111-2222-3333-4444} \and
%Third Author\inst{3}\orcidID{2222--3333-4444-5555}}
%% First names are abbreviated in the running head.
%% If there are more than two authors, 'et al.' is used.
%%
\institute{Purdue University, West Lafayette, USA\\ \email{\{njaber,wagne279,milind,roopsha\}@purdue.edu} \and CISPA Helmholtz Center for Information Security, Saarbr\"{u}cken, Germany \\ \email{jacobs@cispa.saarland}}
%Springer Heidelberg, Tiergartenstr. 17, 69121 Heidelberg, Germany
%ABC Institute, Rupert-Karls-University Heidelberg, Heidelberg, Germany
\maketitle              % typeset the header of the contribution
\begin{abstract}
  Inspired by distributed applications that use consensus or other agreement protocols for global coordination, we define a new computational model for
  parameterized systems that is based on a general global synchronization
  primitive and allows for global transition guards. Our model generalizes many
  existing models in the literature, including broadcast protocols and guarded
  protocols. We show that reachability properties are decidable for systems
  without guards, and give sufficient conditions under which they remain
  decidable in the presence of guards. Furthermore, we investigate cutoffs for
  reachability properties and provide sufficient
  conditions for small cutoffs in a number of cases that are inspired by our
  target applications.
\end{abstract}
\section{Introduction}

Distributed applications are notoriously difficult to implement and reason
about, primarily due to the combinatorial explosion of behaviors resulting from
the interleaving of computation and communication.  Naturally, they have
received a lot of attention from the formal methods community to facilitate
reasoning about correctness properties that are too complex to reason about
informally or manually~\cite{Alur.AutomaticSynthesisDistributed.X.2017,Wilox.Verdi.PLDI.2015,Sergey.ProgrammingProvingDistributed.POPL.2017,Padon.IvySafetyVerification.PLDI.2016,damian.communicationclosed.CAV.2019,Konnov.ShortCounterexampleProperty.X.2017,Hawblitzel.Ironfleet.SOSP.2015,Scalas.VerifyingMessagePassing.PLDI.2019,Gleissenthall.Pretend.Synchrony.POPL.2019,berkovits2019verification,Damm.AutomaticCompositionalSynthesis.X.2014}.

One of the main challenges in {\em fully automated} reasoning about a
distributed system is {\em scalability} in a critical system parameter---the
number of processes---with the epitome of success being {\em parameterized
verification of correctness}---correctness that holds regardless of this
parameter.  Unfortunately, the parameterized verification problem is known to
be undecidable even in very simple cases, for example, finite-state processes
that pass a $2$-valued token in a
ring~\cite{Suzuki.PMCP_UndecidableFirstPPR.1988.PPR}.  Hence, approaches for
parameterized verification are divided into two groups: (i) ones that support 
a large class of systems, but only provide semi-decision
procedures~\cite{abdulla2016parameterized,KaiserKW10} and (ii) ones that
provide fully automatic decision procedures for a well-defined class of
systems, but need to carefully restrict this class of systems to obtain such a
strong result.  While the former cannot provide any guarantee of success, the
latter are often not sufficiently general to model practical examples.

In this work, we target fully-automated parameterized verification for a
significantly more general class of systems than addressed in prior work (cf.
the surveys~\cite{emerson2003model,BloemETAL15,Esparza16}).  Inspired by distributed applications
that use consensus or other {\em agreement protocols} for global coordination, we
introduce \emph{global synchronization protocols}, a new computational model for distributed systems that generalizes 
most of the existing models
based on process synchronization, including models based on pairwise
rendezvous~\cite{GS92}, asynchronous rendezvous~\cite{DelzannoRB02},
negotiation~\cite{ED13} and broadcasts~\cite{EsparzaFM99}.  We show that
despite this generality, we can still decide parameterized verification for
safety properties.  Going beyond that, we show that under certain conditions,
our model can be augmented with global transition guards---which allow
to model semaphore-based access control as well as preconditions for
global consensus-like coordination---while retaining decidability.  This
makes our model one of the most expressive models for which the parameterized
verification problem is still decidable.  
Furthermore, we present several results on
{\em cutoffs} for our model, \ie, the number of processes sufficient to prove
or disprove properties of a parameterized system. 
%We show that cutoffs are, in
%general, at least quadratic in the local state space of a process---even for the
%subclass that corresponds to the broadcast model of Esparza et
%al.~\cite{EsparzaFM99}.
%and therefore we need to consider systems that
%satisfy additional conditions in order to obtain cutoffs that are useful in
%practice. 
Inspired both by the decision procedure and by negative examples
that require large cutoffs, we define sufficient conditions on systems 
in our computational model that make small, practical cutoffs possible.
Finally, we evaluate our approach on several distributed applications, 
showing that they can indeed be modeled as global synchronization protocols, 
and we illustrate the significance of our cutoff results in the verification 
of these benchmarks.

\paragraph{Motivating Example.}
Our system model is inspired by applications that {\em use} agreement
protocols, like leader election or consensus, as building blocks to achieve a
more complex overall functionality. We are interested in a compositional
verification setting where we \emph{assume} that the agreement protocols have
been verified separately and want to {\em guarantee} the overall correctness of
an application without having to explicitly model and verify the agreement
protocols within the application; in particular, we focus on 
a setting where verified agreement protocols are encapsulated 
into an abstraction with precondition obligations and postcondition guarantees. 

Thus, our system model needs to be able to incorporate such 
pre- and postconditions of agreement protocols. As a 
simple example, consider the smoke detector application in 
\figref{MappedIntroExample} whose intended behavior is as follows. Upon detecting smoke, the processes coordinate to choose (up to) 2 processes to report the smoke to the fire department. It uses different types of transitions, several of which are popular in the
literature and are supported by existing decidability results: 
an {\em internal transition} (from state \st{\Env} to state \st{\Ask}), a {\em
broadcast} (on action \tr{\Smoke}), and a {\em negotiation}, i.e., a
synchronous transition of all processes with no distinguished sender (on action
\tr{\Resett}). However, additionally our application requires that some
transitions can only happen under certain conditions, given by guards
$\guard_i$ in transition labels.  For example, action \tr{\Resett} should only
be possible if all processes are in $\guard_3$, i.e., in states \st{\Report} or
\st{\Idle}. And most importantly, in state \st{\Pick} we want the system to
\emph{agree} on (up to) 2 processes that move into state \st{\Report} . This
requires a novel type of transition that we have not found in existing literature,
allowing two processes to take a distinguished role while all other processes
are treated uniformly. To faithfully model agreement of processes, we also
require a guard on this transition, since any agreement protocol is based on
the assumption that all processes are ready (i.e., their local state satisfies
some condition) before invocation of the protocol.

\begin{figure*}[t]
\centering
\begin{minipage}{0.99\linewidth}
\centering
%\begin{tikzpicture}[->,>=stealth',auto,node distance=8cm,semithick,scale=0.6,every node/.style={transform shape}]
\begin{tikzpicture}[>=stealth',auto,node distance=8cm,semithick,scale=0.5,every node/.style={transform shape}]
\tikzstyle{every to}=[font=\fontsize{13}{0}\selectfont, draw = black, align=center]
\tikzstyle{every edge}=[->,font=\fontsize{13}{0}\selectfont, draw = black, align=center]
\tikzstyle{every state}=[minimum size=2cm, font=\fontsize{14}{0}]
\tikzstyle{optionalEdge}=[draw=realGreen]
\tikzstyle{optionalNode}=[text=realGreen]
\tikzstyle{hole}=[rectangle,draw=red,dashed,thick,yshift=3]

%%% node (states) definition%%%%%%
\node (init) at (-2,0) {};
\node[state]	(Env)			at (0,0)	    	{\st{\Env}};
\node[state]	(Ask)			at (5,0)	    	{\st{\Ask}};
\node[state]	(Pick)			at (10,0)	    	{\st{\Pick}};
\node[state]	(Report)		at (16,0)	    	{\st{\Report}};
\node[state]	(Idle)			at (10,-4)	    	{\st{\Idle}};

%%% paths definition%%%%%%

\draw
(init) edge (Env)
(Env)			edge										node		[distance=1cm] 										{$G_1$}										(Ask)
(Env)			edge										node		[above, rotate =-21.8, xshift=10, yshift=3]	{ $\trr{\Smoke}$??}				(Idle)
(Ask)			edge	[bend left] 						node		[above] 													{$\trr{\Smoke}!!, G_1$}		(Pick)
(Ask)			edge	[bend right]						node		[above, yshift=5] 									{$\trr{\Smoke}??$}				(Pick)
(Pick)			edge										node		[right] 													{$\trr{\Choose}??$}				(Idle)
(Pick)			edge	[bend left]						node		[above] 													{$\trr{\Choose_1}!!, G_2$}	(Report)
%(Idle)		edge	[loop right]						node		[] 													 		{$\trr{\Choose}??}				(Idle)
(Pick)			edge	[bend right]						node		[above, yshift=5] 									{$\trr{\Choose_2}!!, G_2$}	(Report)
(Report.90)	edge	[bend right]						node		[above, yshift=3] 									{$\trr{\Resett}, G_3$} 			(Env)
(Idle)			edge	[bend left]						node		[above, rotate =-21.8, yshift=1] 					{$\trr{\Resett}, G_3$} 			(Env)
(Report)     node [below =1.9cm,font=\fontsize{16}{0}\selectfont, xshift = -5]		{$G_1=\{\st{Env},\st{Ask}\}$}
(Report)     node [below =2.7cm,font=\fontsize{16}{0}\selectfont, xshift = -0.3]	{$G_2=\{\st{Pick},\st{Idle}\}$}
(Report)     node [below =3.5cm,font=\fontsize{16}{0}\selectfont, xshift = 11.95]	{$G_3=\{\st{Report},\st{Idle}\}$}
;
  
\end{tikzpicture}     
\end{minipage}

\caption{A smoke detector process. The internal transition from initial state
  \st{\Env} to \st{\Ask} models that a process detects smoke (an environment signal). A process that
  detected smoke can initiate a broadcast \tr{\Smoke}, moving all processes
  from \st{\Env} to \st{\Idle} and from \st{\Ask} to \st{\Pick}, where the
  transition \tr{\Choose} moves (up to) 2 processes to \st{\Report}, and the
  rest from \st{\Pick} to \st{\Idle}.   %If the \consagree from \st{Pick} is a maximal sender transition, then up to 2 processes may move to \st{Report}. If it is a multi-sender transition , then exactly 2 must move to \st{Report}. Optionally, we may consider just one of the actions (say, $b_1!!$, without $b_2!!$ present) as a simple broadcast, and just one process would move to \st{Report}.
Finally, all processes from \st{\Report} and \st{\Idle} may move back to \st{\Env} in a synchronous transition with no dedicated sender. Transitions labeled with a set $\guard_i$ can only be taken if all processes are in this set.
The safety property for a distributed smoke detector based on this process is that at most 2 processes should report the fire.}
\label{fig:MappedIntroExample}

\end{figure*}

\section{System Model: Global Synchronization Protocols}
\label{sec:prelim} \label{sec:basicModel}

We present \emph{global synchronization protocols} (GSPs), 
a formal system model that generalizes most of the existing 
synchronization-based models in the literature~\cite{EsparzaFM99,GS92,DelzannoRB02,ED13}, 
including models based on rendezvous and broadcasts. 
In this model, each global transition synchronizes all processes, 
where an arbitrary number $k$ of processes \emph{act} as the senders of the 
transitions, while the remaining processes react
 uniformly as receivers.
The model supports two basic types of transitions: 
(i) a \emph{$k$-sender transition}, which \emph{can fire} only if at least $k$
 processes are ready to act as senders, and \emph{is fired} with exactly $k$ 
processes acting as senders, and (ii) 
a \emph{$k$-maximal transition}, which \emph{can fire} if the number $m$ of 
processes that are ready to act as senders is at least $1$, and \emph{is 
fired} with $min(m,k)$ processes acting as senders.
%\mn{was originally: ``and \emph{is 
%fired} with $k$ processes acting as senders if $m \geq k$, or with $m$ 
%processes acting as senders, otherwise.''}
%\mnn{switched from max to min}
Additionally, each transition can be equipped with a \emph{global guard} that 
identifies a subset of the local state space. Then, a transition is 
\emph{enabled} whenever it can fire and the local states of all processes are in 
the set identified by the transition guard.

We formalize these notions in the following, starting with the case without 
transition guards.
%\footnote{Note that for presentation clarity, we do not explicitly consider an  environment process in our model. All of our results extend to the case with an explicit environment process; see~\appref{gbc-env}\ for a justification.}
%\mn{Combined footnote here with footnote 2, to avoid introducing E before system processes are defined.}

\subsection{Global Synchronization without Guards}
\label{sec:unguardedGSPs}
\paragraph{Unguarded Processes.} An \emph{unguarded process} is a labeled 
transition system 
$P=\tuple{A,\LS,\ls_0,T}$, where $A$ is a set of \emph{local actions}, $\LS$ 
is a finite set of states, $\ls_0 \in \LS$ is the initial state, and 
$T \subseteq \LS \times A \times \LS$ is the transition relation. 
$A$ is based on a set $\GActions$ of \emph{global actions}, where each 
$a \in \GActions$ has an \emph{arity} $k \geq 1$ and is either a \emph{$k$-sender action} or a \emph{$k$-maximal action}. For every global action $a \in \GActions$ with arity $k$, $A$ 
contains \emph{local actions} $a_1{!!},\ldots,a_k{!!},a{??}$. Actions $a_1{!!},\ldots,a_k{!!}$ are called \emph{sending actions} and $a{??}$ is called a \emph{receiving action}.

A local transition from state $\ls$ to state $\ls'$ on sending action $\alpha \in A$ 
denoted $\trans{\ls}{\alpha}{\ls'}$ is called a \emph{sending transition} 
(resp., \emph{receiving transition}) if $\alpha$ is a sending action (resp., receiving 
action).
We assume that receives are \emph{deterministic}: for each state $\ls$ and each receiving action $a{??}$, there is 
exactly one state $\ls'$ with $\trans{\ls}{a{??}}{\ls'}$, and that sends are \emph{unique}: for each sending action 
$a_i$ there is exactly one pair of states $\ls, \ls'$ with $\trans{\ls}{a_i{!!}}{\ls'}$.
\footnote{Processes that do not satisfy the assumptions can easily be rewritten to satisfy them, e.g. by adding self-loops on any missing receive actions, and by renaming the actions of duplicate sending transitions (and adding corresponding receiving transitions).} %\mnnn{Should we change the footnote to say ``See the extended version of a justification'' and remove the appendix? or just remove the appendix and give no justification?}

\begin{example}
If we ignore guards on transitions, the process in \figref{MappedIntroExample} is an unguarded process. Global action \textbf{Choose} has arity $2$, and local sending transitions $\trans{\st{\Pick}}{\textbf{Choose}_i!!}{\st{\Report}}$ for $i \in \{1,2\}$. One local receiving transition is $\trans{\st{\Pick}}{\textbf{Choose}??}{\st{\Idle}}$, and all other receiving transitions on \textbf{Choose} are self-loops (not depicted).
\end{example}

\paragraph{Unguarded Systems.} %\mn{\roopsha{Maybe combine this with the previous paragraph?}}
Given an unguarded process $P=\tuple{A,\LS,\ls_0,T}$,
we consider systems composed of $n$ identical processes, and use a counter 
abstraction to efficiently represent global states, without loss of 
precision~\cite{EmersonT99}.\footnote{For presentation clarity, we do not explicitly consider an 
environment process in our model. All of our results extend to the case with 
an explicit environment process; see \appref{gbc-env} for a justification.}

That is,
the parameterized global transition system is defined as 
$\M(n)=\tuple{\GActions,Q,\conf_0,\rightarrow}$, where
$Q = \{0,\ldots,n\}^{\LS}$, i.e., a global state is a function $\conf: \LS \rightarrow \{0,\ldots,n\}$. Assuming a fixed order on $\LS$, we will also use $\conf$ as a vector of natural numbers. The initial state $\conf_0$ is the state with $\conf_0(s_0)=n$ and $\conf_0(s)=0$ for all $s\neq s_0$.
Finally, we define the global transition relation $\rightarrow$, separated into the two different types of actions:
%\footnote{Note that full symmetry of our systems, as shown in Lemma~\ref{}, is a precondition for this abstraction of the state space to be complete.}

\noindent \emph{$k$-sender actions.} A $k$-sender action $a \in \GActions$ with local sending transitions $\trans{\ls_i}{a_i!!}{\ls_i'}$ for $i \in \{1,\ldots,k\}$ can be fired from a global state $\conf$ if there are $k$ processes that can take these local transitions.
			Upon firing the action, each of the local transitions on actions $a_i!!$ is taken by exactly one process, and all other processes take a 
			transition on action $a??$ to arrive in the new global state $\conf'$.
Formally, we assign to each $k$-sender action $a \in \GActions$ (i) a vector
  $\textbf{v}_a \in Q$ containing the number of
  expected senders for each state $t \in \LS$: $\textbf{v}_a(t)= | \{
  \trans{\ls}{a_i{!!}}{\ls'} \mid \ls = t \}|$, (ii) a vector
  $\textbf{v}_a'$ containing the number of senders that will be in each state $t \in \LS$ after the
  transition: $\textbf{v}_a'(t)= | \{ \trans{\ls}{a_i{!!}}{\ls'} \mid \ls' = t \}|$, and (iii) a function $M_a: \LS \times \LS \rightarrow \{0,1\}$, where $M_a(\ls,\ls')=1$ if there is a local transition $\trans{\ls}{a??}{\ls'}$, and $M_a(\ls,\ls')=0$ otherwise. We also use $M_a$ as a $|\LS| \times |\LS|$ matrix, called the \emph{synchronization matrix} of action $a$.
%	
  %For instance, assuming that the global action $\trr{Choose}$ in
  %\figref{MappedIntroExample} is a $2$-sender action, then
  %$\textbf{v}_{\trr{\Choose}}(\st{\Pick}) = 2$ and
  %$\textbf{v}'_{\trr{\Choose}}(\st{\Report}) = 2$.

Then, a transition from global state $\conf$ on action $a$ is possible if $\conf(\ls_i) \geq \textbf{v}_a(\ls_i)$ for all $i \in \{1,\ldots,k\}$, and the resulting global state can be computed as 
$$\conf' = M_a \cdot (\conf - \textbf{v}_a) + \textbf{v}_a',$$
			and we write $\trans{\conf}{a}{\conf'}$. Intuitively, $\conf'$ is obtained from $\conf$ by ``removing'' the senders from their local start states, moving all the remaining (receiving) processes to their respective local destination states, and then adding the senders to their appropriate local destination states.
	Note that this representation relies on the assumption that sends are unique and receives are deterministic, which also implies that each column of a synchronization matrix $M_a$ is a unit vector.

\begin{example}
\label{ex:unguarded}
%\mn{style edits?}
Consider the process in \figref{MappedIntroExample}. The synchronization matrix and vectors for action \tr{\Smoke} are shown below, with global states given in the order 
%$\tuple{\st{\textbf{E}nv} ,\ \st{\textbf{A}sk}, \ \st{\textbf{I}dle}, \ \st{\textbf{P}ick}, \ \st{\textbf{R}eport}}$
$\tuple{\st{Env} ,\ \st{Ask}, \ \st{Idle}, \ \st{Pick}, \ \st{Report}}$ (and abbreviated as $\tuple{\st{E} ,\ \st{A}, \ \st{I}, \ \st{P}, \ \st{R}}$).
Notice, for instance, that the first column in $M_{\trr{\Smoke}}$ encodes the local receive transition $\trans{\st{\Env}}{\trr{\Smoke}??}{\st{\Idle}}$. %On the other hand, t
The vector-pair $\textbf{v}_{\trr{\Smoke}}$ and $\textbf{v}'_{\trr{\Smoke}}$ encode the local send transition $\trans{\st{\Ask}}{\trr{\Smoke}!!}{\st{\Pick}}$. In particular, $\textbf{v}_{\trr{\Smoke}}$ indicates that the sender starts in \st{\Ask} and $\textbf{v}'_{\trr{\Smoke}}$ indicates that the sender moves to \st{\Pick}.
\[
  \begin{array}{C{\mycw}}
        \st{E} \\ \st{A} \\ \st{I} \\ \st{P} \\ \st{R} \\
  \end{array}
  \stackrel{
  	\stackrel{
  		\mbox{$M_{\trr{\Smoke}}$}
	}{\begin{array}{C{\mycw}C{\mycw}C{\mycw}C{\mycw}C{\mycw}}
        \st{E} & \st{A} & \st{I} & \st{P} & \st{R} \\
  \end{array}}
  }{
  \left[ {\begin{array}{C{\mycw}C{\mycw}C{\mycw}C{\mycw}C{\mycw}}
        0 & 0 & 0 & 0 & 0 \\
        0 & 0 & 0 & 0 & 0 \\
        1 & 0 & 1 & 0 & 0 \\
        0 & 1 & 0 & 1 & 0 \\
        0 & 0 & 0 & 0 & 1 \\
  \end{array} } \right]
  }
  \quad
  \begin{array}{C{\mycw}}
        \st{E} \\ \st{A} \\ \st{I} \\ \st{P} \\ \st{R} \\
  \end{array}
  \hspace{-1em}
  \stackrel{\stackrel{\mbox{$\textrm{\textbf{v}}_{\bf \Smoke}$}}{\begin{array}{C{\mycw}}   \\\end{array}}}{
  \left[ {\begin{array}{C{\mycw}}
        0 \\ 1 \\ 0 \\ 0 \\ 0 \\
  \end{array} } \right]
   }
   \quad
   \begin{array}{C{\mycw}}
        \st{E} \\ \st{A} \\ \st{I} \\ \st{P} \\ \st{R} \\
   \end{array}
   \hspace{-1em}
   \stackrel{\stackrel{\mbox{$\textrm{\textbf{v}}'_{\bf \Smoke}$}}{\begin{array}{C{\mycw}}   \\\end{array}}}{
  \left[ {\begin{array}{C{\mycw}}
        0 \\ 0 \\ 0 \\ 1 \\ 0 \\
  \end{array} } \right]
   }
\]

\noindent Now, consider a global state $\tuple{3, 2, 0, 0, 0}$ with three processes in \st{\Env} and two in \st{\Ask}. From this state, the transition $\trans{\tuple{ 3, 2, 0, 0, 0}}{\trr{~\Smoke}}{\tuple{ 0, 0, 3, 2, 0}}$ is enabled (since there is at least 1 sender in $\st{Ask}$), where all three processes in \st{\Env} act as receivers to move to \st{\Idle} (according to the synchronization matrix $M_{\trr{\Smoke}}$), one process in \st{\Ask} acts as the sender to move to \st{\Pick}, and the other process in \st{\Ask} acts as a receiver, also moving to \st{\Pick}. 

\end{example}

\noindent \emph{$k$-maximal actions.}
A $k$-maximal action $a \in \GActions$ with local sending transitions $\trans{\ls_i}{a_i!!}{\ls_i'}$ for $i \in \{1,\ldots,k\}$ can be fired from a global state $\conf$ if there is at least one process that can take one of these local transitions.
Upon firing the action, for each state $\ls_i$ with at least one local transition 
			$\trans{\ls_i}{a_i!!}{\ls_i'}$, (i) if $\conf(\ls_i) \geq \textbf{v}_a(\ls_i)$ then each of the local transitions $\trans{\ls_i}{a_i!!}{\ls_i'}$ is taken by exactly one process, or, (ii) if $\conf(\ls_i) < \textbf{v}_a(\ls_i)$ then a total of $\conf(\ls_i)$ of the local transitions $\trans{\ls_i}{a_i!!}{\ls_i'}$ are taken, each by exactly one process.
			All other processes 
			take a transition on the receiving action $a??$ to arrive in the new global state $\conf'$.
			Formally, we again assign to each action $a$ vectors $\textbf{v}_a, \textbf{v}'_a$ and a synchronization matrix $M_a$, as above. If $\conf(\ls_i) \geq \textbf{v}_a(\ls_i)$ for all $i \in \{1,\ldots,k\}$, then these are used as defined above. For cases where this does not hold, we assign to the action an additional set of vector-pairs $(\textbf{u}_a,\textbf{u}'_a)$ with different numbers of senders that actually participate, and $\conf'$ is computed based on a vector-pair with the maximal number of senders that is supported by $\conf$.
			%\swen{long version of formalization in comments}
			
			%Formally, we assign to a $k$-maximal action $a \in \GActions$ a vector $\textbf{v}_a$ for the pre-states of expected senders, a vector $\textbf{v}'_a$ for the post-states of the senders, and a synchronization matrix $M_a$, which are defined exactly as for $k$-sender transitions. However, the new state is only computed as 
			%$$\conf' = M_a \cdot (\conf - \textbf{v}_a) + \textbf{v}_a'$$
%if $\conf(\ls_i) \geq \textbf{v}_a(\ls_i)$ for all $i \in \{1,\ldots,k\}$. For cases where $\conf(\ls_i) < \textbf{v}_a(\ls_i)$ for one or more $i \in \{1,\ldots,k\}$, we additionally assign to $a$ a set of pairs of vectors $(\textbf{u}_a,\textbf{u}'_a)$ such that (i) $\textbf{u}_a \leq \textbf{v}_a$, (ii) $\textbf{u}'_a \leq \textbf{v}_a'$, and (iii) $\textbf{u}_a(\ls_i) \geq 1$ for at least one $i \in \{1,\ldots,k\}$, such that these pairs cover all cases where not all senders are present. Then, if $\conf(\ls_i) < \textbf{v}_a(\ls_i)$ for one or more $i \in \{1,\ldots,k\}$, the new state is computed as 
			%$$\conf' = M_a \cdot (\conf - \textbf{u}_a) + \textbf{u}_a'$$
			%for some $(\textbf{u}_a,\textbf{u}'_a)$, with $\textbf{u}_a(\ls_i)=\conf(\ls_i)$ for all $i \in \{1,\ldots,k\}$ with $\conf(\ls_i)<\textbf{v}_a$. Note that for a given $\textbf{u}_a$, there may be two pairs $(\textbf{u}_a,\textbf{u}'_a)$ and $(\textbf{u}_a,\textbf{u}''_a)$ with $\textbf{u}'_a \neq \textbf{u}''_a$ that both correspond to valid $k$-maximal transitions.
						%We write $\trans{\conf}{a}{\conf'}$.}

\begin{example}
\label{ex:unguarded2}
%%%%%%%%%%%%%%%%%%%%%%%%%%%%%%%%%%%%
The synchronization matrix and vectors for action \tr{\Choose} are shown below. Note that, if \tr{\Choose} is a $2$-maximal action, then the vector-pair $(\textbf{u}_{\trr{\Choose}},$ $\textbf{u}'_{\trr{\Choose}})$ is used to model
the case where only one sender is available to take the sending transition.
\[
  \begin{array}{C{\mycw}}
        \st{E} \\ \st{A} \\ \st{I} \\ \st{P} \\ \st{R} \\
  \end{array}
  \stackrel{\stackrel{\mbox{$M_{\trr{\Choose}}$}}{\begin{array}{C{\mycw}C{\mycw}C{\mycw}C{\mycw}C{\mycw}}
        \st{E} & \st{A} & \st{I} & \st{P} & \st{R} \\
  \end{array}}}{
  \left[ {\begin{array}{C{\mycw}C{\mycw}C{\mycw}C{\mycw}C{\mycw}}
        1 & 0 & 0 & 0 & 0 \\
        0 & 1 & 0 & 0 & 0 \\
        0 & 0 & 1 & 1 & 0 \\
        0 & 0 & 0 & 0 & 0 \\
        0 & 0 & 0 & 0 & 1 \\
  \end{array} } \right]
  }
  \quad
   \begin{array}{C{\mycw}}
        \st{E} \\ \st{A} \\ \st{I} \\ \st{P} \\ \st{R} \\
   \end{array}
   \hspace{-1em}
  \stackrel{\stackrel{\mbox{$\textrm{\textbf{u}}_{\bf \Choose}$}}{\begin{array}{C{\mycw}}   \\\end{array}}}{
  \left[ {\begin{array}{C{\mycw}}
        0 \\ 0 \\ 0 \\ 1 \\ 0 \\
  \end{array} } \right]
   }
   \quad
    \begin{array}{C{\mycw}}
        \st{E} \\ \st{A} \\ \st{I} \\ \st{P} \\ \st{R} \\
   \end{array}
   \hspace{-1em}
   \stackrel{\stackrel{\mbox{$\textrm{\textbf{u}}'_{\bf \Choose}$}}{\begin{array}{C{\mycw}}   \\\end{array}}}{  \left[ {\begin{array}{C{\mycw}}
        0 \\ 0 \\ 0 \\ 0 \\ 1 \\
  \end{array} } \right]
   }
   \quad
   \begin{array}{C{\mycw}}
        \st{E} \\ \st{A} \\ \st{I} \\ \st{P} \\ \st{R} \\
   \end{array}
   \hspace{-1em}
   \stackrel{\stackrel{\mbox{$\textrm{\textbf{v}}_{\bf \Choose}$}}{\begin{array}{C{\mycw}}   \\\end{array}}}{  \left[ {\begin{array}{C{\mycw}}
        0 \\ 0 \\ 0 \\ 2 \\ 0 \\
  \end{array} } \right]
   }
   \quad
   \begin{array}{C{\mycw}}
        \st{E} \\ \st{A} \\ \st{I} \\ \st{P} \\ \st{R} \\
   \end{array}
   \hspace{-1em}
   \stackrel{\stackrel{\mbox{$\textrm{\textbf{v}}'_{\bf \Choose}$}}{\begin{array}{C{\mycw}}   \\\end{array}}}{  \left[ {\begin{array}{C{\mycw}}
        0 \\ 0 \\ 0 \\ 0 \\ 2 \\
  \end{array} } \right]
   }
\]

\noindent Regardless of whether \tr{\Choose} is a $2$-sender or a $2$-maximal action, the global transition $\trans{\tuple{ 0, 0, 1, 4, 0}}{\trr{\Choose}}{\tuple{ 0, 0, 3, 0, 2}}$ is possible. In a state $\conf = \tuple{ 0, 0, 4, 1, 0}$, with 4 processes in \st{\Idle} and 1 in \st{\Pick}, the \tr{\Choose} action will not be enabled if it is a $2$-sender action because two sending processes are required (in \st{\Pick}), but only one sender is available. However, if \tr{\Choose} is a $2$-maximal action, then the global transition $\trans{\tuple{ 0, 0, 4, 1, 0}}{\trr{\Choose}}{\tuple{ 0, 0, 4, 0, 1}}$ is possible.

\end{example}

\paragraph{Runs, Reachability Properties.}
A \emph{run} of system $\M(n)$ is a finite or infinite sequence of global 
states $\conf_0 \conf_1 \ldots$, where $\conf_0$ is the initial state and 
$\trans{\conf_i}{a}{\conf_{i+1}}$ for all $i$. We say that a state $\conf$ 
is \emph{reachable} in $\M(n)$ if there is a run of $\M(n)$ that ends in 
$\conf$.
For a fixed $m \in \Nat$ and local state $\ls \in \LS$, let $\phi_m(\ls)$ be a property denoting the 
reachability of a global state $\conf$ with $\conf(\ls)\geq m$. If such a 
state is reachable in $\M(n)$, we write $\M(n) \models \phi_m(\ls)$.
%\mn{\roopsha{Why is this interesting/useful?}}

\paragraph{Other Communication Primitives in the \gbc Model.}
%\mn{add informal defs to all?}
Note that most of the synchronization-based communication primitives from the 
literature are instances of $k$-sender transitions or $k$-maximal 
transitions: \emph{broadcasts}~\cite{EsparzaFM99} are simply $1$-sender transitions, 
\emph{internal transitions} are $1$-sender transitions with $M_a=Id$ (the 
identity matrix), \emph{pairwise rendezvous transitions}~\cite{GS92} are 
$2$-sender transitions (denoting the sender and receiver of the rendezvous transition) with $M_a=Id$, \emph{asynchronous rendezvous transitions}~\cite{DelzannoRB02} are $2$-maximal transitions with $M_a=Id$.
%\emph{negotiations}~\cite{ED13} are $0$-sender transitions.
\emph{Negotiations}~\cite{ED13}, i.e., a synchronous transition of all processes with
no distinguished sender, can be modeled as a set of $1$-sender 
transitions, where every local receiving transition $\trans{\ls}{a??}{\ls'}$ 
is paired with a sending transition $\trans{\ls}{a!!}{\ls'}$, allowing an 
arbitrary process to act as the sender. 
%\mn{\swen{introduce special notation for internal transitions and negotiations?}}
In addition to these, GSPs allow us to express many other natural 
synchronization primitives, e.g., summarizing the election of (up to) $k$ 
leaders in a single step.

Finally, \emph{disjunctive guards}~\cite{EK00}, i.e., 
guards $\guard \subseteq \LS$ that require that there \emph{exists} a process 
that is in some state $\ls \in \guard$, can be modeled by adding an auxiliary 
sending action $a_{\guard}!!$, and transitions $\trans{\ls}{a_{\guard}!!}{M_a(
\ls)}$ for 
every $\ls \in \guard$, i.e., a process in some state $\ls \in \guard$ must 
exist to enable the 
transition, but apart from that this process acts like a receiver. Note that this works 
without adding a notion of guards to our model.

In what follows, we extend our model to allow \emph{conjunctive guards}, i.e., guards that require that \emph{all} processes are in some subset of the local state space.

\subsection{Global Synchronization with Guards}

\paragraph{Guarded Processes.} 
%\label{sec:extendedmodel}
A \emph{guarded process} is a tuple
$\ph=\tuple{A,\LS,\ls_0,T}$, where all components are as before, except that 
now we have $T \subseteq \LS \times A \times \guards \times \LS$, i.e., 
transitions are additionally labeled with a subset of $\LS$, called a \emph{guard}. 
A local transition from state $\ls$ to state $\ls'$ on action $\alpha$ with 
guard $\guard$
will be denoted $\trans{\ls}{\alpha,\guard}{\ls'}$.
We call a guard $\guard$ \emph{non-trivial} if $\guard \neq \LS$.
Wlog, we assume that for any global action $a$, all local transitions based on $a$
have the same guard.
 
%\noindent \emph{Composition of processes with guards.}

\paragraph{Guarded Systems.}
Let the \emph{support} of a global state $\conf$ be $\supp(\conf) = \{ \ls \in \LS 
\mid \conf(\ls) > 0 \}$, i.e., the set of local states that appear at least 
once in $\conf$. %\nour{might need to extend this to also be over bags}
Then the semantics of a global transition on action $a$ with 
guard $G$, denoted $\trans{\conf}{a, \guard}{\conf'}$, is as defined before, 
except that the transition is enabled only if $\supp(\conf) \subseteq G$.

\begin{example}
Consider the global transitions introduced in \exref{unguarded}, and recall that global states are given in the order $\tuple{\st{\Env} ,\ \st{\Ask}, \ \st{\Idle}, \ \st{\Pick}, \ \st{\Report}}$. 
While the transition $\trans{\tuple{ 0, 0, 1, 4, 0}}{\trr{~\Resett~}}{\tuple{ 1, 0, 0, 4, 0}}$ would be possible in the unguarded model, the guard $G_3 = \{\st{\Report},\st{\Idle}\}$ on the \tr{\Resett} action disables this transition, as $\supp(\tuple{ 0, 0, 1, 4, 0})  = \{\st{\Pick},\st{\Idle}\} \not \subseteq G_3$.
Similarly, from $\conf = \tuple{1 , 0 , 1 , 2 ,0}$, while a transition on action ${\bf \Choose}$ is enabled for unguarded processes, the guard $G_2=\{\st{Pick},\st{Idle}\}$ on action \tr{\Choose} disables this transition, since $\supp(\tuple{1 , 0 , 1 , 2 ,0}) \not \subseteq G_2$.
%
%Furthermore, in a state 
% as there are 4 processes in \st{\Pick} $\notin G_3$.
%
\end{example}

\section{Parameterized Verification for GSPs without Guards}
\label{sec:parametrizedSafety}

In this section, instead of the parameterized system $\M(n)$,
we consider an infinite-state system $\M_\infty$ that includes the behaviors 
of $\M(n)$ for
every $n$: it initializes to $\M(n)$ for arbitrary $n \in \Nats$, and then
behaves according to the semantics of a GSP of that size. We are 
interested in reachability properties $\phi_m(\ls)$, where $\M_\infty 
\models \phi_m(\ls)$ is equivalent to $\exists n.~ \M(n) \models \phi_m(s)$, 
i.e., we are considering a \emph{parameterized reachability property} over 
all instances of $\M$.

We use this slightly different
model in order to make use of the notion of \emph{well-structured transition
  systems (WSTS)}, as defined by Finkel \cite{Finkel87}: an infinite-state 
transition
system that is equipped with a \emph{well-quasi-order (WQO)} on its state space
and has some additional properties. Finkel and Schnoebelen~\cite{FS01} have 
surveyed existing results on WSTSs and put them into a common
framework. 

We will show that, for a suitable WQO, $\M_\infty$ 
is a WSTS, and that this enables parameterized verification for reachability 
properties $\phi_m(\ls)$.

%guarded broadcast protocols are WSTSs, and that we can decide the parameterized verification problem 
%parameterized model checking problem (PMCP) 
%for safety properties. \\

%\subsection{The Well-Quasi Order}
%
%Given the set $C \subset Q$ that defines the partition of local states, we define the WQO on global states as 
%$$\conf \wqo \conf' \text{ iff } \left(\conf(q) \leq \conf'(q) \text{ for all } q \in Q \text{ and } \sum_{q \notin C}\conf(q)=0 \text{ iff } \sum_{q \notin C}\conf(q)=0\right).$$
%
%That is, $\conf'$ must contain at least as many processes as $\conf$ in every local state, and the two configurations must agree on whether all processes are in $C$.

\subsection{Compatibility and Effective Computability of Predecessors}
%To prove decidability of the PMCP,
For the following definitions, fix an infinite set of states $\GS$ and a transition relation 
$\rightarrow$. Moreover, let $\wqo$ be a WQO on $\GS$, i.e., a reflexive and transitive relation such that, for any infinite sequence $\conf_0, \conf_1, \conf_2, \ldots$ of states from $\GS$, there exist indices $i < j$ with $\conf_i \wqo \conf_j$. In particular, $\wqo$ does not admit infinitely decreasing sequences or infinite anti-chains.

\paragraph{Compatibility.} We say that $\wqo$ is \emph{compatible} with
$\rightarrow$ if for every $\conf, \conf', \altconf \in \GS$ with $\conf \wqo
\altconf$ and $\conf \rightarrow \conf'$ there exists $\altconf' \in \GS$ with
$\conf' \wqo \altconf'$ and $\altconf \rightarrow^* \altconf'$. If the 
property also holds after replacing $\altconf \rightarrow^* \altconf'$ with 
$\altconf \rightarrow \altconf'$, then we say $\wqo$ is \emph{strongly 
compatible} with $\rightarrow$.

\paragraph{Well-Structured Transition System.} A transition system $(\GS,\rightarrow)$
equipped with a WQO that is compatible with $\rightarrow$ is called a
\emph{well-structured transition system (WSTS)}.

\paragraph{Upwards-Closed Sets.}
For a (possibly infinite) subset $U \subseteq \GS$, the \emph{upwards closure} of $U$ is the set
$\uparrow U = \{ \altconf \in \GS \mid \exists \conf \in U: \conf \wqo \altconf\}$. A set $U$ is \emph{upwards closed} if $\uparrow U = U$. Every 
upwards closed set $U$ has a finite \emph{basis}: a finite set $B \subseteq U$
 such that $\uparrow B = U$. 
%We also denote the basis of $U$ by $basis(U)$.

\paragraph{Effectively Computable Predecessors.}  
For $U \subseteq \GS$, let $Pred(U)$ denote the 
predecessor states of $U$ with respect to $\rightarrow$. 
We say that we can \emph{effectively compute $Pred$} if there exists an algorithm that computes a finite basis of $Pred(U)$ from any finite basis of any upwards-closed $U \subseteq \GS$.

\begin{theorem}[\cite{FS01}]
In a WSTS with effectively computable $Pred$, reachability of any upwards-closed set is decidable.
\end{theorem}

\subsection{Decidability for Unguarded GSPs}
\label{sec:decUGSPs}
We prove that any unguarded GSP is a WSTS with effectively computable $Pred$, 
which implies that reachability properties are decidable for GSPs. 
To this end, let $\wqo$ be the component-wise order on global state vectors 
$\conf$, $\altconf$:
$$\conf \wqo \altconf \; \text{ iff } \; \conf(\ls) \leq \altconf(\ls) \text{ for 
all } \ls \in \LS.$$

Note that with respect to this WQO, the set of global states $\conf$ with 
$\conf(\ls) \geq m$ is an upwards-closed set, i.e., if we can decide 
reachability of upwards-closed sets, then we can decide reachability 
properties $\phi_m(\ls)$. Thus, decidability
of checking $\M_\infty \models \phi_m(\ls)$
%\new{of parameterized verification for reachability properties $\phi_m(\ls)$}
 follows from the following 
theorem.

%We will show that, for a suitable WQO, $\M_\infty$ 
%is a WSTS, and that this enables parameterized verification for reachability 
%properties $\phi_m(\ls)$.

\begin{theorem}
\label{thm:generalizedBC} 
If $\M_\infty$ is based on an unguarded GSP process, then
$\M_\infty$ equipped with $\wqo$ is a WSTS and we can effectively compute Pred.
\end{theorem}

\begin{proof}
To prove that $\M_\infty$ is a WSTS, we show strong compatibility of 
transitions w.r.t. $\wqo$. %, i.e., \mnnn{\chris{is this redefinition of strong compatibility needed?}\swen{remove if we need space, otherwise it is a help for the reader}}
%if $\conf \wqo \altconf$ and $\conf \rightarrow \conf'$, 
%then $\exists \altconf'$ with $\conf' \wqo \altconf'$ and $\altconf 
%\rightarrow \altconf'$.
We consider the following two cases separately: (i) $k$-sender 
transitions, and (ii) $k$-maximal transitions.

(i) For $k$-sender transitions, 
let $\conf \wqo \altconf$ and $\trans{\conf}{a}{\conf'}$ for some 
$k$-sender action $a$. Then $\conf' = M_a \cdot (\conf - \textbf{v}_a) + \textbf{v}_a'$ for some 
synchronization matrix $M_a$ and vectors $\textbf{v}_a, \textbf{v}_a'$ associated with action $a$. 
First observe that since $\conf \wqo \altconf$, there is also a transition $\trans{\altconf}{a}{\altconf'= M_a \cdot (\altconf - \textbf{v}_a) + \textbf{v}_a'}$.
Moreover, we have $M_a \cdot \conf \wqo M_a \cdot \altconf$, and 
therefore $M_a \cdot (\conf - \textbf{v}_a) + \textbf{v}_a' \wqo M_a \cdot (\altconf - \textbf{v}_a) + \textbf{v}_a'$, i.e., $\conf' \wqo \altconf'$.

(ii) For $k$-maximal transitions, consider again $\conf \wqo \altconf$ and 
$\trans{\conf}{a}{\conf'}$, where now $a$ is a $k$-maximal action.
Then $\conf' = M_a \cdot (\conf - \textbf{u}_{a,\conf}) + \textbf{u}_{a,\conf}'$ for some vectors $\textbf{u}_{a,\conf},\textbf{u}_{a,\conf}'$ with 
$\sum_{s \in S}\mathbf{u}_{a,\conf}(s)=\sum_{s \in S}\mathbf{u}_{a,\conf}'(s) \leq k$. Again, first observe that since $\conf \wqo \altconf$, a transition  
$\trans{\altconf}{a}{\altconf'}$ is enabled, where $\altconf' = M_a \cdot 
(\altconf - \textbf{u}_{a,\altconf}) + \textbf{u}_{a,\altconf}'$ and $\textbf{u}_{a,\altconf}(\ls) \geq \textbf{u}_{a,\conf}(\ls)$, $\textbf{u}_{a,\altconf}'(\ls) \geq \textbf{u}_{a,\conf}'(\ls)$ 
for all $\ls \in \LS$. Note that, for any $\ls \in \LS$, we can have 
$\textbf{u}_{a,\altconf}(\ls) > \textbf{u}_{a,\conf}(\ls)$ only if $\conf(\ls)-\textbf{u}_{a,\conf}(\ls) \leq 0$ and $\altconf(\ls) > 
\conf(\ls)$. Furthermore, $\textbf{u}_{a,\altconf}(\ls) - \textbf{u}_{a,\conf}(\ls) \leq \altconf(s) - \conf(s)$. 
%\mnnn{You cannot get any more senders in the bigger system than the number of processes added (since you have already used all the possible senders in the smaller system)}
Therefore, we get $\conf - \textbf{u}_{a,\conf} \wqo \altconf - \textbf{u}_{a,\altconf}$, which implies  
$M_a \cdot (\conf - \textbf{u}_{a,\conf}) \wqo M_a \cdot (\altconf - \textbf{u}_{a,\altconf})$, and thus
$M_a \cdot (\conf - \textbf{u}_{a,\conf}) + \textbf{u}_{a,\conf}' \wqo M_a \cdot (\altconf - \textbf{u}_{a,\altconf}) + \textbf{u}_{a,\altconf}'$, i.e., $\conf' \wqo \altconf'$.

Next, we prove that we can effectively compute the basis of $Pred(C)$, where $Pred(C)$ is the set of states from which a transition exists to a state in an upwards-closed set $C$, as follows:

(i) For a $k$-sender transition based on action $a$, any predecessor $\conf$ in $Pred(C)$ must satisfy (i) $\textbf{v}_a \wqo \conf$, and (ii) $M_a \cdot (\conf - \mathbf{v}_a) + \textbf{v}_a' = \conf'$, for some $\conf' \in C$. The basis of $Pred(C)$ consists of the minimal elements (w.r.t. $\wqo$) that satisfy these conditions, and thus is computable.%\mn{this still seems okay. I think we should either (1) use this minimality as a way to ensure finite basis from the infinite set $Pred(C)$, or state that the basis of the pred(C) is elements that can reach the basis of C etc.}

%\nour{old:(i) For a $k$-sender transition based on action $a$, then, any predecessor $\conf$ that reaches $Pred(C)$ through $a$ must satisfy (i) $\textbf{v}_a \wqo \conf$, and (ii) $M_a \cdot (\conf - \mathbf{v}_a) + \textbf{v}_a' = \conf'$, for some $\conf' \in C$. The basis of $Pred(C)$ consists of the minimal elements (w.r.t. $\wqo$) that satisfy these conditions, and thus is computable.}

(ii) For $k$-maximal transitions, the proof works in the same way,
 except that now we may have multiple possibilities of what a minimal 
predecessor could be, based on different subsets of the senders being present 
or not. Since this is always a finite case distinction, effective 
computability of $Pred$ is still guaranteed.
\qed
\end{proof}

\section{Parameterized Verification for GSPs with Guards}
\label{sec:PV-guards}
For GSPs with guards, compatibility under $\wqo$ in general 
does not hold, since for $\conf \wqo \altconf$, a transition 
on action $a$ that is enabled in $\conf$ may not be enabled in $\altconf$.
Furthermore, note that even strong restrictions on processes are unlikely to yield 
compatibility with respect to $\wqo$, since whenever $\supp(\conf)\subseteq \guard$ for a non-trivial $\guard$, one can always find a $\altconf$ 
with $\conf \wqo \altconf$ and $\supp(\altconf) \nsubseteq \guard$, disabling the action. 
%\mn{old: Furthermore, note that even strong restrictions on processes are unlikely to yield compatibility with respect to $\wqo$, since for any given $\conf$ that 
%satisfies a non-trivial guard $\guard \neq \LS$ one can always find a $\altconf$ 
%with $\conf \wqo \altconf$ such that $\altconf$ does not satisfy $\guard$, i.e., such that $\altconf(\ls) > 0 $ for some local state $\ls \notin \guard$.} 

Therefore, we introduce a refined WQO, denoted $\altwqo$, that is based on the semantics of 
guards, as well as sufficient conditions on the guarded process $P$, 
such that the system $\M_\infty$ is a WSTS and we can effectively compute $
Pred$.

Let $\mathcal{G}$ be the set of guards that appear on transitions in $P$, and recall that $\supp(\conf) = \{ \ls \in \LS \mid \conf(\ls) > 0 \}$. 
Then we consider the following WQO\footnote{We show that 
$\altwqo$ is a WQO by proving that every infinite sequence of global states $\conf_1, 
\conf_2, \ldots$ contains $\conf_i, \conf_j$ with $i
<j$ and $\conf_i \altwqo \conf_j$. To this end, consider an arbitrary infinite 
sequence $\overline{\conf} = \conf_1, \conf_2, \ldots$. Then there is at 
least one set $S$ of local states such that infinitely many 
$\conf_i$ have $\supp(\conf_i)=S$. Let $\overline{\conf'}$ be the 
infinite 
subsequence of $\overline{\conf}$ where all elements have $\supp(\conf'_i)=S$. 
Since $\wqo$ is a WQO, there exist $\conf'_i, \conf'_j$ with $i<j$ and 
$\conf'_i \wqo \conf'_j$, and since $\supp(\conf'_i)=\supp(\conf'_j)=S$, we 
also get 
$\conf'_i \altwqo \conf'_j$. Since $\conf'_i = \conf_k$ and $\conf'_j = \conf_l$
 for some $k < l$, we get $\conf_k \altwqo \conf_l$ for $k < l$, and thus 
$\altwqo$ is a WQO.}:
%\mn{\swen{need to show that this is a WQO?}}
%\footnote{To see that 
%this is really a WQO, we need to show that every infinite sequence $\conf_1, 
%\conf_2, \ldots$ of global states contains contains $\conf_i, \conf_j$ with $i
%<j$ and $\conf_i \altwqo \conf_j$. To see this, consider an arbitrary infinite 
%sequence $\overline{\conf} = \conf_1, \conf_2, \ldots$. Then there is at 
%least one set $S$ of local states such that there are infinitely many 
%$\conf_i$ with $\supp(\conf_i)=S$. Then let $\overline{\conf'}$ be the 
%infinite 
%subsequence of $\overline{\conf}$ where all elements have $\supp(\conf'_i)=S$. 
%Since $\wqo$ is a WQO, there exist $\conf'_i, \conf'_j$ with $i<j$ and 
%$\conf'_i \wqo \conf'_j$, and since $\supp(\conf'_i)=\supp(\conf'_j)=S$, we 
%also get 
%$\conf'_i \altwqo \conf'_j$. Since $\conf'_i = \conf_k$ and $\conf'_j = \conf_l$
% for some $k < l$, we get $\conf_k \altwqo \conf_l$ for $k < l$, and thus 
%$\altwqo$ is a WQO.}
$$\conf \altwqo \altconf \; \text{ iff } \; \left( \conf \wqo \altconf \land \forall \guard \in \mathcal{G}: \left( \supp(\conf) \subseteq \guard \iff \supp(\altconf) \subseteq \guard \right) \right).$$

%\mn{this is still true with the ``iff'' maybe let's keep it as is?}
Intuitively, a global state $\altconf$ is considered 
greater than a global state $\conf$ if $\altconf$ has at least as many 
processes as $\conf$ in any given state, \emph{and} for every transition 
$\trans{\conf}{a}{\conf'}$ that is enabled in 
$\conf$, a transition on action $a$ is also enabled in $\altconf$. 

We will see that compatibility with respect to $\altwqo$ can only be ensured under 
additional conditions, as formalized in the following.
%We note that this requirement essentially states that the parameterized system moves in rounds (if one thinks of each broadcast as a round), and that the boundaries of such rounds (i.e. where those broadcasts are enabled or disabled) is independent of the specific instantiation of the parameterized system.}

\subsection{Guard-Compatibility and Well-Behaved Processes}
\label{sec:well-behavedness}
%We start by defining the notion of strong guard-compatibility, which implies strong compatibility of transitions with respect to $\altwqo$. Then, we relax our conditions to weak guard-compatibility, which only implies compatibility w.r.t. $\altwqo$.\mnnn{Might rephrase since we downplayed the weak guard compatibility.}

\paragraph{Strong Guard-Compatibility for $k$-Sender Actions.}
For a $k$-sender action $a$ with local sending transitions $\trans{\ls_i}{a_i!!, \guard}{\ls_i'}$ for $i \in \{1,\ldots,k\}$, let $\senderset$ be the set of all states $\ls_i$, $\senderset'$ the set of states $s_i'$, and $M_a$ the synchronization matrix. 
We say that action $a$ is 
\emph{strongly guard-compatible} if the following holds for all $\guard' \in \guardset{
:}$

%\left(\senderset \subseteq \guard' \land \senderset' \subseteq \guard' 
%\Rightarrow \forall \ls \in \guard{:}~ \left(\ls \in \guard' \Rightarrow M_a(
%\ls) \in \guard'\right) \right)\\
%\land \left(\senderset \not \subseteq \guard' \land 
\[ 
\senderset' \subseteq \guard' \Rightarrow \forall \ls \in \guard{:}~ M_a(\ls) \in \guard' 
\tag{C1} \label{eq:guard-comp-k-sender}
\]
%\right)

Intuitively, if all senders move into a guard $\guard'$, then also all 
receivers need to move into $\guard'$. This ensures that if $\guard'$ is satisfied after the transition in a system of a given size, then it is satisfied after that transition in a system of any bigger size, because any additional receivers must also move into $\guard'$. Note that Condition~\eqref{eq:guard-comp-k-sender} always holds for trivial guards.

%Furthermore, note that negotiation actions are a special case with $k=0$. Because there are no senders, the premise of ~\eqref{eq:guard-comp-k-sender} always holds, so negotiations are strongly guard-compatible only when all associated receive transitions end within (or outside of) the same guards.

%\mn{\swen{need the special case for rendezvous?}}
%For the special case of rendezvous transitions (where $k=2$ and $M_a=Id$), 
%$\trans{\ls_i}{a_1!!, \guard}{\ls_i'}$ 
%and $\trans{\ls_j}{a_2!!, \guard}{\ls_j'}$, we get that action $a$ is \emph{guard-compatible} if 
%for all $\guard' \in \mathcal{G}$:

%\[
%\left( \left( 
%\ls_i' \in \guard' \land \ls_j' \in \guard'\right) \Rightarrow 
%\left( \guard \subseteq \guard' \right) \right)
%\tag{C1a} \label{eq:guard-comp-rendezvous}
%\]

\paragraph{Strong Guard-Compatibility for $k$-Maximal Actions.}
%For a $k$-maximal action $a$, the idea of the condition is the same as 
%before. However, since the action can fire with different subsets of available senders, we need to ensure that all senders must agree, for every $\guard \in 
%\guardset$, on whether they enter $\guard$ or not.
For a $k$-maximal action $a$ with local transitions $\trans{\ls_i}{a_i!!,G
}{\ls_i'}$ for $i \in \{1,\ldots,k\}$, as before, let $\senderset$ be the set of all states $\ls_i$, $\senderset'$ the set of states $s_i'$, and $M_a$ the synchronization matrix. 
We say that action $a$ is 
\emph{strongly guard-compatible} if the following holds for all $\guard' \in \guardset{
	:}$
\[ 
\left(\senderset' \cap \guard' \neq \varnothing \right) \Rightarrow \left( \senderset' \subseteq \guard' \wedge \forall \ls \in \guard{:}~ M_a(\ls) \in \guard' \right)
\tag{C2} \label{eq:guard-comp-k-max}
\]
%\editt{or}
%\[ 
%\left(\bigvee_{1 \leq i \leq k} \ls_i' \in \guard' \right)
% \Rightarrow \left( \senderset' \subseteq \guard' \wedge \forall \ls \in \guard{:}~ M_a(\ls) \in \guard' \right)
%\tag{C2} \label{eq:guard-comp-k-max}
%\]

%\[
%\left(\bigvee_{1 \leq i \leq k} \ls_i' \in \guard' \right) \Rightarrow \left( \forall \ls \in R: M_a(\ls) \in \guard' \right)
%\tag{C2.1} \label{eq:guard-comp-k-max-1}
%\]
%\[
%\bigwedge_{i,j \in \{1,\ldots,k\}} \left( (\ls_i \triangleleft \ls_j \land \ls_j' \in \guard') \Rightarrow (\ls_i' \in \guard' \land M_a(\ls_i) \in \guard')
%\right)
%\tag{C2.2} \label{eq:guard-comp-k-max-2}
%\]

Intuitively, if one potential sender moves from a state $\ls_j$ into a guard $\guard'$, then every other sender and receiver must do the same, so that $\guard'$ will be satisfied regardless of the number of receivers, or the subset of senders that actually participated in the action

Note that for $k=1$, the condition~\eqref{eq:guard-comp-k-max} is equivalent to 
condition~\eqref{eq:guard-comp-k-sender}. This is to be expected, since semantically there 
is no difference between a $1$-sender action and a $1$-maximal action.

\begin{example}
We can see that actions \tr{\Smoke}, \tr{\Choose}, and \tr{\Resett} from our motivating example in \figref{MappedIntroExample} are strongly guard-compatible:
\begin{compactenum}[--]

\item \tr{\Smoke} is a 1-sender action with sending transition \trans{\st{\Ask}}{\trr{\Smoke}!!, \{\st{\Env, \Ask}\}}{\st{\Pick}}. The state \st{\Pick} is only included in one non-trivial guard $G_2$ = \{\st{\Pick}, \st{\Idle}\}. Since receiving transitions from \{\st{\Env, \Ask}\} end in $\{\st{\Pick, \Idle}\} \subseteq G_2$, condition~\eqref{eq:guard-comp-k-sender} holds, so \tr{\Smoke} is strongly guard-compatible.

\item Consider \tr{\Choose} with sending transitions $\trans{\st{\Pick}}{\trr{\Choose_i!!},\{\st{\Pick,\Idle}\}}{\st{\Report}}$ for $i \in \{1,2\}$ as a 2-sender action. \st{\Report} is only included in one non-trivial guard $G_3$ = \{\st{\Report}, \st{\Idle}\}. Since the receiving transition from \{\st{\Pick}\} ends in $\st{\Idle} \in G_3$ as well,~\eqref{eq:guard-comp-k-sender} holds, so \tr{\Choose} is strongly guard-compatible.

%\item Consider \tr{\Choose} with sending transitions $\trans{\st{\Pick}}{\trr{\Choose_i!!},\{\st{\Pick,\Idle}\}}{\st{\Report}}$ for $i \in \{1,2\}$ as a 2-sender action. \st{\Report} is only included in one non-trivial guard $G_3$ = \{\st{\Report}, \st{\Idle}\}. Since the receiving transition from \{\st{\Pick}\} ends in $\st{\Idle} \in G_3$ as well,~\eqref{eq:guard-comp-k-sender} holds, so \tr{\Choose} is strongly guard-compatible.
\item Consider \tr{\Choose} as a 2-maximal action. Again, \st{\Report} is only included in one non-trivial guard $G_3$ = \{\st{\Report}, \st{\Idle}\}. Since all senders and receivers start from \st{\Pick} and end up in a state in $G_3$, condition \eqref{eq:guard-comp-k-max} holds and \tr{\Choose} is, again, strongly guard-compatible.

%\item Consider \tr{\Choose} as a 2-maximal action. Conditions \eqref{eq:guard-comp-k-max-1} and \eqref{eq:guard-comp-k-max-2} are implied by~\eqref{eq:guard-comp-k-sender}, so because~\eqref{eq:guard-comp-k-sender} holds, \tr{\Choose} is, again, strongly guard-compatible.

% in the 2-sender case, because both sending actions move from $\st{\Pick}$ to $\st{\Report}$. \st{\Report} is only included in one non-trivial guard $G_3$ = \{\st{\Report}, \st{\Idle}\}. Because the receiving transition from \{\st{\Pick}\} ends in $\st{\Idle} \in G_3$ as well, condition~\eqref{eq:guard-comp-k-sender} is satisfied. 

%, and for non-trivial guards this action trivially satisifes conditions \eqref{eq:guard-comp-k-max-1} and \eqref{eq:guard-comp-k-max-2} by negating the premises of all implications (and for the trivial guard $\guard=\LS$ all the consequences are trivially satisfied).

\item \tr{\Resett} is a negotiation action. Recall that negotiations are modeled as a set of 1-sender actions, allowing for an arbitrary sender. Therefore, each of these broadcasts must satisfy~\eqref{eq:guard-comp-k-sender} for the negotiation to be guard-compatible.
\tr{\Resett} is indeed strongly guard-compatible because all of its sending and receiving transitions end in \st{Env}, meaning that when the action fires, all processes will move into a single state, ensuring that all guards will be uniformly enabled or disabled, regardless of the number of processes, which of them is the sender, or whether they begin in \st{Report} or \st{Idle}.

\item Finally, as stated in \secref{unguardedGSPs}, the internal transition
$\trans{\st{\Env}}{G_1}{\st{\Ask}}$ can be modeled by a 1-sender action, say $a$, with a send transition $\trans{\st{\Env}}{a!!,G_1}{\st{\Ask}}$ and self-loop receive transitions on all states. The sender ends up in one non-trivial guard $G_1=\{\st{\Env},\st{\Ask}\}$. Since receiving transitions from $\{\st{\Env, \Ask}\}$ end in $\{\st{\Env, \Ask}\} \subseteq G_1$, condition~\eqref{eq:guard-comp-k-sender} holds, so $a$ is strongly guard-compatible.
\end{compactenum}
\end{example}

\paragraph{Refinement: Weak Guard-Compatibility.} 
To support a larger class of transitions, we show how to relax the previous 
conditions, at the cost of making them more complex. 
The idea is that, instead of requiring that any guard that is 
satisfied by the senders is also satisfied by the receivers immediately after 
the transition, it is enough if 
the receivers \emph{have a path} to a state that satisfies all these guards. To 
avoid unnecessary complexity, we only consider paths of internal transitions.

We write $t \triangleleft \ls$ if, for all guards $\guard \in \guardset$, $s 
\in \guard \Rightarrow t \in \guard$. Similarly, we write $t \triangleleft H$ for a set of states $H$ if, for all guards $\guard \in \guardset$, $H \subseteq \guard \Rightarrow t \in \guard$. If there exists a path of unguarded internal transitions from $\ls$ to $\ls'$, we write 
$\ls \rightsquigarrow \ls'$.
Then, condition~\eqref{eq:guard-comp-k-sender} can be relaxed to
\[ 
\senderset' \subseteq \guard' \Rightarrow \forall \ls \in \guard{:}~ \left( M_a(\ls) \in \guard' \lor \exists \ls' \in \LS: (\ls' \triangleleft \senderset' \land M_a(\ls) \rightsquigarrow \ls') \right) 
\tag{C1w} \label{eq:guard-comp-k-sender-weak}
\]

%\edit{old}
%In a similar way, conditions~\eqref{eq:guard-comp-k-max-1} can be relaxed to obtain condition (C2.1w)  by replacing $M_a(\ls) \in \guard'$ with $M_a(\ls) \in \guard' \lor \exists \ls' \in \LS: (\forall \ls_i': (\ls' \triangleleft \ls_i') \land M_a(\ls) \rightsquigarrow \ls')$, and (C2.2w) is obtained from \eqref{eq:guard-comp-k-max-2} by replacing $M_a(\ls_i) \in \guard'$ with $M_a(\ls_i) \in \guard' \lor \exists \ls' \in \LS: (\ls' \triangleleft \ls_i' \land M_a(\ls_i) \rightsquigarrow \ls')$.
%
%
%\edit{new}
In a similar way, condition~\eqref{eq:guard-comp-k-max} can be relaxed to obtain condition (C2w)  by replacing $M_a(\ls) \in \guard'$ with $M_a(\ls) \in \guard' \lor \exists \ls' \in \LS: (\forall \ls_i': (\ls' \triangleleft \ls_i') \land M_a(\ls) \rightsquigarrow \ls')$.

Actions that satisfy condition \eqref{eq:guard-comp-k-sender-weak} or (C2w) are called \emph{weakly guard-compatible}.
%\mn{\swen{not sure if this is needed:}}
%Note that it is in general not enough to require that $t'$ satisfies $\guard'$: since we are quantifying over all guards in our condition, the weaker condition might be satisfied if there are different paths to different states $t'$ that individually satisfy one guard that $\ls'$ satisfies, but still no path to a global state $\altconf'$ with $\conf' \altwqo \altconf'$.

\paragraph{Refinement: Internal Transitions.}
For internal transitions, we can relax the condition a bit further. Intuitively, since all receiving processes stay in their current state, it is enough to require condition (C1w) to hold if a guard is satisfied after the transition that was not satisfied before. Furthermore, the internal transitions that the ``receivers'' can use to reach a state that satisfies the same guards as the sender can be guarded, as long as the guard is general enough to guarantee that these transitions can be taken.\footnote{Guarding these internal transitions is also possible with more complicated conditions in other cases, but for simplicity we only introduce it formally for internal transitions.}

%\footnote{In fact, guarding these internal transitions is also possible under more complicated conditions in other cases, but for simplicity we only introduce it formally for internal transitions.}

Formally, we write $\ls \rightsquigarrow_G \ls'$ if there exists a path of internal transitions from $\ls$ to $\ls'$ such that each transition is guarded by some $\guard' \supseteq \guard$. Then, we say that an internal action $a$ is \emph{weakly guard-compatible} if
the following holds for every $\guard' \in \guardset$ (note that $\senderset$ 
and $\senderset'$ are both singleton sets, denoted $\ls$ and $\ls'$ in the 
following) and every $t \in \guard$:
\[
\ls \notin \guard' \land \ls' \in \guard' \Rightarrow \exists t':~ (t' \triangleleft \ls' \land t\rightsquigarrow_{\guard''} t'),
\tag{C3w} \label{eq:guard-comp-internal}
\]
where $\guard'' = \guard \cup \{ t \mid t \triangleleft \ls'\}$.

\paragraph{Well-Behavedness.}
Based on guard-compatibility, we can now 
define the class of processes that will allow us to retain decidability of 
reachability properties in the parameterized system:
We say that a process $P$ is \wellbehaved if every action is (weakly) guard-compatible.

Note that unguarded processes are trivially well-behaved.

\begin{example}
Observing that all actions in the process depicted in~\figref{MappedIntroExample} are (strongly) guard-compatible, it is clear that the process is \wellbehaved.
\end{example}

\paragraph{Well-Behaved Systems in the Literature.}
We want to point out that many systems studied in the literature are naturally well-behaved.

For example, Emerson and Kahlon~\cite{EK03a} introduce a model for cache 
coherence protocols that is based 
on broadcast communication and guards. They show that many textbook protocols 
can be modeled under the following restrictions:
(i) every state is assumed to have an unguarded internal transition to the 
initial state $\st{Init}$, and (ii) the only conjunctive guard is 
$\{ \st{Init} \}$. Clearly, every action in a process that satisfies these 
conditions will also satisfy condition~\eqref{eq:guard-comp-k-sender-weak}, 
and therefore well-behaved systems subsume and significantly generalize 
the types of protocols considered by Emerson and Kahlon. 

Moreover, there has recently been much research on the verification of round-based distributed systems~\cite{Gleissenthall.Pretend.Synchrony.POPL.2019,damian.communicationclosed.CAV.2019,choosepaper.arxiv}, where processes can move 
independently to some extent, with the restriction that transitions between rounds
can only be done synchronously for all processes. When abstracting from certain features (e.g. fault-tolerance and process IDs), our model is well-suited to express such 
systems: guards can be used to restrict transitions to happen only in a 
certain round, and can furthermore model the ``border'' of a round that needs 
to be reached by all processes, such that they can jointly move to the next 
round.

Our example from \figref{MappedIntroExample} can also be seen 
as a round-based system: the first round includes states $\st{\Env, \Ask}$, and 
upon taking the transition on $\trr{\Smoke}$, all processes move to the second 
round, which includes states $\st{\Pick, \Idle}$. From there, on action 
$\trr{\Choose}$ the system moves to the third round, which includes states 
$\st{\Report, \Idle}$, and on action $\trr{\Resett}$ back to the first round. 
Note that the states in different rounds are exactly the guards that are used 
in the transitions---or seen the other way around, guards induce a set of 
rounds on the local state space, and the guard-compatibility conditions 
ensure that processes move between these rounds in a systematic way.

While the rounds are very simple in this example, the technique is much more 
general and can be used to express many round-based systems, including 
those described in~\secref{evaluation}.

%Naturally, the design of these systems ensures that when an action is sent, then all other processes are in phases

%If a specific action can be sent in round x, then all other processes can receive it during that round.

%There is an expected set of input/output actions that processes utiliz ewhich are available in a particular set.

%This may seems restrictive, but a large class of dist. systems are round-based....

%In the lifetime of dsitributed systems, processes are usually in subsets
%They're kind of moving in different hases, and even though technically you can be anywhere, the execution flow is somewhat...

\subsection{Decidability for Well-Behaved Guarded Processes}

Based on the notion of well-behavedness, we can now obtain a decidability 
result that works in the presence of guards. The following theorem implies 
that parameterized verification for properties $\phi_m(\ls)$ is decidable 
for well-behaved processes.

\begin{theorem}
\label{thm:GPS-with-guards}
If $\M_\infty$ is based on a well-behaved GSP process, then
$\M_\infty$ is a WSTS and we can effectively compute Pred.
\end{theorem}

\begin{proof}
To prove that $\M_\infty$ is a WSTS, we show compatibility of 
transitions w.r.t. $\altwqo$, i.e., 
if $\conf \altwqo \altconf$ and $\conf \rightarrow \conf'$, 
then $\exists \altconf'$ with $\conf' \altwqo \altconf'$ and $\altconf 
\rightarrow^* \altconf'$.
We consider three cases: (i) $k$-sender, (ii) $k$-maximal, and (iii) internal transitions.

(i) Suppose $a$ is a $k$-sender action. Let 
$\trans{\conf}{a, \guard}{\conf'}$ be a transition and $\conf \altwqo 
\altconf$.
Since $\conf \altwqo \altconf$ implies that $\supp(\altconf) \subseteq \guard$, we know that transition $\trans{\altconf}{a, \guard}{\altconf'}$ 
is possible, and by the proof of \thmref{generalizedBC} we know that 
$\conf' \wqo \altconf'$. To prove compatibility with respect to $\altwqo$, it 
remains to show that $\forall \guard' \in \guardset: (\supp(\conf')\subseteq 
\guard' \Rightarrow \supp(\altconf')\subseteq \guard')$.

First assume that condition~\eqref{eq:guard-comp-k-sender} holds. 
Then, let $\guard' \in \guardset$ be an arbitrary guard. By \eqref{eq:guard-comp-k-sender}, 
we either have $\senderset \not\subseteq \guard'$,
in which case the 
desired condition is satisfied for $\guard'$, or we have that $\forall \ls \in
 \guard{:}~ M_a(\ls) \in \guard'$, i.e., all potential receivers move into $
\guard'$. Thus, we get $\supp(\conf')\subseteq \guard'$ iff $\supp(\altconf')
\subseteq \guard'$, satisfying the desired condition.

If instead of \eqref{eq:guard-comp-k-sender} the action satisfies \eqref{eq:guard-comp-k-sender-weak}, the argument is the same, except that if necessary we use the internal transitions that are guaranteed to exist by the condition to arrive in a state $\altconf'$ with $\conf' \wqo \altconf'$.
%\mn{Awkward phrasing}

(ii)
Suppose $a$ is a $k$-maximal action with local transitions 
$\trans{\ls_i}{a_i!!,G}{\ls_i'}$ for $i \in \{1,\ldots,k\}$ and synchronization 
matrix $M_a$. By the proof of \thmref{generalizedBC} we know that there exists a transition $\trans{\altconf}{a, \guard}{\altconf'}$ with
$\conf' \wqo \altconf'$, and it remains to show that $\forall \guard' \in \guardset: (\supp(\conf')\subseteq 
\guard' \iff \supp(\altconf')\subseteq \guard')$. 

% new stuff
Let $\guard' \in \guardset$ be an arbitrary guard, and assume the action is strongly guard-compatible.
By condition \eqref{eq:guard-comp-k-max} we know that 
we either have $\senderset \cap \guard' = \varnothing$,
in which case the desired condition is satisfied for $\guard'$, or $\left( \senderset' \subseteq \guard' \wedge \forall \ls \in \guard{:}~ M_a(\ls) \in \guard' \right)$, i.e., all potential senders and receivers move into $
\guard'$. Thus, we get $\supp(\conf')\subseteq \guard'$ iff $\supp(\altconf')
\subseteq \guard'$, satisfying the desired condition.

Again, if the action is weakly guard-compatible, the argument can be extended by using the paths of internal transitions, if necessary.

(iii) 
Suppose $a$ is an internal action with local transition $\trans{\ls}{a,G}{\ls'}$, let 
$\trans{\conf}{a, \guard}{\conf'}$ and $\conf \altwqo \altconf$.

%\swen{To support internal transitions and possibly some other cases, we need to have a form of non-strong compatibility that allows a transition to be ``repeated'' until a local state (or a set of states, for multiple senders) becomes empty. This may be necessary to finally arrive at a state that satisfies the same guards.}
If $\ls' \notin \guard'$ or $\ls \in \guard'$, we can simply take the same local transition from $\altconf$ and trivially, $\supp(\altconf') \subseteq \guard'$. If $\ls' \in \guard'$ and $\ls \notin \guard'$, then by condition~\eqref{eq:guard-comp-internal} every $t \in \guard$ has an executable path to a state $t'$ with $t' \triangleleft \ls'$, and therefore, by definition of $\triangleleft$, $t' \in \guard'$. For each of these states, if $\conf'(t)=0$ and $\altconf(t)>0$, then from $\altconf$ we take this path for all processes in $t$, and arrive in a state $\altconf'$ with $\conf' \altwqo \altconf'$.
%
%For $\guard \in \guardset$ and $\ls, t \in \LS$, write $\ls \leq_\guard t$ if $\ls \notin \guard$ implies $t \notin \guard$. Extended to the set of all guards, write $\ls \leq_\guardset t$ if there exists $\guard$ such that $\ls \leq_\guard t$.

Effective computability of $Pred$ follows 
from the proof of \thmref{generalizedBC}---the only difference is that we 
must consider the guards, i.e., a predecessor is only valid if it 
additionally satisfies the guard of the transition under consideration.
\qed
\end{proof}

\section{Cutoffs for {\gbc}s}
\label{sec:cutoffs}

We investigate cutoff results for {\gbc}s and their 
connection to the decidability results in \thmref{generalizedBC} and 
\ref{thm:GPS-with-guards}. 
While the proofs of these theorems
yield a decision procedure for parameterized verification, a cutoff result is more versatile as it reduces parameterized 
verification to a problem over a fixed number of processes, and under certain conditions can also be used for parameterized synthesis~\cite{Jacobs.ParameterizedSynthesis.X.2012}.

\subsection{Definition and Basic Observations}
%Formally, a \emph{cutoff} for a class of processes $\mP$ and a class of specifications $\Phi$ is a number $c \in \Nats$ such that for every $P \in \mP$ and $\varphi \in \Phi$,
%\[ \forall n \geq c: \left( P^c \models \varphi \iff P^n \models \varphi \right). \] 
%\nour{again, is it up to the cutoff or at the cutoff?}
A \emph{cutoff} for a class of processes $\Pi$ and a class of properties $\Phi$ is a number $c \in \Nat$ such that for every $P \in \Pi$ and $\phi \in \Phi$,
$$ \M_\infty \models \phi \Leftrightarrow \M(c) \models \phi$$

We show how to obtain cutoffs for well-behaved {\gbc}s that satisfy additional conditions, and for reachability properties of the form $\phi_m(\ls)$, based on observations from the proof of \thmref{generalizedBC}.
While for any given parametrized system and any safety property a cutoff exists~\cite{Namjoshi07}, a general cutoff, even if it can be computed, may be too large to be of practical value:
 it has been shown that for broadcast protocols the time complexity of checking reachability is non-primitive recursive in the size of the processes~\cite{SchmitzS13}, and from the proof one can conclude that the same must hold for the size of cutoffs.

%Note that this idea applies more broadly than only to the cutoff results that we will give in the following, and even more broadly than our \wellbehaved systems. It may be useful for any WSTS with broadcast communication and a WQO that is based on $\preceq$.

\paragraph{Example: Quadratic Cutoffs.}
%We give an example that shows that, even for simple reachability properties, there are \gbc protocols that have a cutoff that is quadratic in the size of the process.
%
%\begin{figure}[t]
%\centering
%\includegraphics[width=\textwidth]{quadratic_cutoff.png}
%%\caption{Example witnessing quadratic cutoff}
%%\label{fig:cutoff-example}
%\end{figure}
%
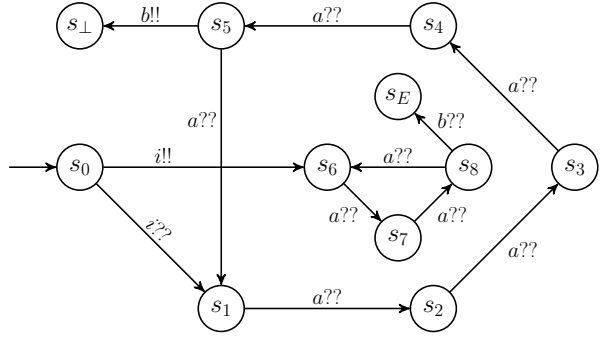
\begin{figure}[t]
\centering
\begin{tikzpicture}[>=stealth',auto,node distance=8cm,semithick,scale=0.47,every node/.style={transform shape}]

\tikzstyle{every edge}=[font=\fontsize{16}{0}\selectfont, draw = black, align=center,->]
\tikzstyle{every state}=[minimum size=1.3cm, font=\fontsize{20}{0}]

%%% node (states) definition%%%%%%

\coordinate     (c1)		at (-4,0);
\node[state]	(s0)		at (-2,0)	    		{$s_0$};
\node[state]	(s1)		at (-2,4)	    		{$s_\bot$};
\node[state]	(s2)		at (5,0)	    		{$s_6$};
\node[state]	(s3)		at (9,0)	    		{$s_8$};
\node[state]	(s4)		at (12,0)	    		{$s_3$};

\node[state]	(s5)		at (7,2)	    		{$s_E$};
\node[state]	(s6)		at (7,-2)	    		{$s_7$};

\node[state]	(s7)		at (8,4)	    		{$s_4$};
\node[state]	(s7')		at (2,4)	    		{$s_5$};
\node[state]	(s8)		at (2,-4)	    		{$s_1$};
\node[state]	(s9)		at (8,-4)	    		{$s_2$};

%%% paths definition%%%%%%
\draw

(c1)	edge	(s0)
(s7')	edge	node	[above] 	{$b!!$}			(s1)
(s0)	edge	node	[pos=0.3] 	{$i!!$}			(s2)
(s0)	edge	node	[above, rotate = -45] 	{$i??$}		(s8)
(s2)	edge	node	[below left] 	{$a??$}			(s6)
(s6)	edge	node	[below right] 	{$a??$}			(s3)
(s3)	edge	node	[above] 	{$a??$}			(s2)
(s8)	edge	node	[] 	{$a??$}			(s9)
(s9)	edge	node	[below right] 	{$a??$}			(s4)
(s3)	edge	node	[above right] 	{$b??$}			(s5)
(s4)	edge	node	[above right] 	{$a??$}			(s7)
(s7)	edge	node	[above] 	{$a??$}			(s7')
(s7')	edge	node	[left, pos=0.3] 	{$a??$}			(s8)

%(c1.60)     edge   node [above, rotate = 7.1] {$a_1$}  (w1.120)
%(c1)    edge  node [above, rotate = 7.1] {$a_2$}  (w1.200)
%(c1.-40)     edge  node [above, rotate = -7.1] {$a\#$}  (l1.200)
%(c1)     node [above =1.7cm,font=\fontsize{20}{0}\selectfont, xshift = 40] {$G=\{c_1,c_2,x\}$}
;
\end{tikzpicture}
\caption{Example witnessing quadratic cutoff. Not depicted are additional sending transitions on $a!!$ from every state in the outer cycle to $s_\bot$.}
\label{fig:cutoff-example}
\end{figure}
Consider the (unguarded) process in \figref{cutoff-example}. 
%\mn{\swen{make this more prominent, into a Lemma/Thm?}}
We are 
interested in a lower bound on the cutoff for this process, with respect to 
$\phi_1(s_E)$, i.e., reachability of $\ls_E$ by at least one process. 
Note that to reach $\ls_E$, we need at least one process in $\ls_8$ and one 
in $\ls_5$ at the same time.
From the initial state $s_0$, the only possible action is $i$, sending one 
process to $\ls_6$ in the inner cycle and all other processes to $\ls_1$ in 
the outer cycle.
Then, the only way to make progress is action 
$a$, moving the process in the inner cycle to $\ls_7$, the sending 
process from $\ls_1$ to $\ls_\bot$ (sending transitions on $a!!$ are not 
depicted in \figref{cutoff-example}), and all other processes to 
$\ls_2$. 
After three further transitions on $a$, the outer processes are in $\ls_5$, 
where the sending transition on $b!!$ could be fired, but the process in 
the inner cycle is in $\ls_7$, so additional transitions on $a$ are required. 
Only after two additional rounds around the outer cycle we arrive in a 
state where both $\ls_5$ and $\ls_8$ are occupied, and we can take the final 
transition on $b$ that takes one process into $\ls_E$. To arrive 
there, we took $16$ transitions (one on $i$, $14$ on $a$, and one on $b$), and 
by construction every process can only take one sending transition in a 
run. Thus, we need a system with at least $16$ processes 
to have one of them reach $\ls_E$, and no smaller number can be a 
cutoff for $\phi_1(s_E)$. 

To see that cutoffs grow at least quadratically, note that in similar examples where 
the inner and outer
cycles consist of $p_1$ and $p_2$ states, respectively, and $p_1$ and $p_2$ 
are relatively prime, then we need $p_1 \cdot p_2 + 1$ processes to reach $\ls_E$. 

%Finally, note that this process only uses $1$-sender transitions (i.e., 
%broadcasts), and no guards. That is, 
%the quadratic lower bound on cutoffs not only holds for {\gbc}s, but 
%also for the broadcast protocols of Esparza et al.~\cite{EsparzaFM99}.

\subsection{Conditions for Small Cutoffs}
%\mnnn{Minor rewording, please double-check.} 
%\mn{It looks like this section still freely uses negotiations, internal transitions, broadcasts etc. 
  %These needs to made consistent with previous sections which only use $k$-sender and $k$-maximal transitions.}
We introduce sufficient conditions on processes that allow us to obtain small cutoffs. These conditions are inspired by our intended applications (see \secref{evaluation}), and based on insights from the decision procedure in the proof of \thmref{generalizedBC} and the example above. We observe that any $\conf \in Pred(C)$ that reaches a state $\conf' \in C$ through a $k$-sender action $a$ must satisfy (i) $\textbf{v}_a \wqo \conf$, and (ii) $M_a \cdot (\conf - \mathbf{v}_a) + \textbf{v}_a' = \conf'$. Thus, if there is $\conf' \in C$ such that $\neg (\textbf{v}_a \wqo \conf)$, we need to consider a predecessor $\conf$ with $|\conf| > |\conf'|$. 
It is easy to see that this can only happen if $\conf'$ contains processes in states that can be reached through $a$ only through either a receiving transition, or a sending transition if $k>1$. Thus, we want to avoid that states we are interested in are only reachable through such transitions.

We restrict our attention to specifications $\phi_m(\ls)$
%for some $m\in \Nats$, and 
and to cases where we can identify conditions on a GSP process $P$ such that the cutoff for such specifications is $c=m$.
If this is the case, then we say that reachability of $\ls$ is \emph{\nice} in $P$, and that the pair  $\tuple{P,\phi_m(\ls)}$ is \amenable.

We begin with a simple case, where systems are restricted to only internal transitions and negotiations (we defined in Section~\ref{sec:unguardedGSPs} how these are expressed in terms of $1$-sender transitions).

\begin{lemma}
\label{lem:cutoffnegotiation}
Let $P=\tuple{A,\LS,\ls_0,T}$ be a \wellbehaved GSP process such that all transitions are internal transitions or negotiations. Then reachability of $s$ is \nice in $P$ for every $\ls \in \LS$. 
\end{lemma}

\begin{proof}
To see this, first consider a system with $n>m$ processes, where eventually 
$m$ of them reach $\ls$. We can simulate this run in a system with $m$ 
processes by simply keeping the $m$ processes that reach $\ls$, and removing 
all others. Similarly, if all processes in a system of size $m$ eventually 
reach $\ls$, then we can simulate this run in a bigger system by adding 
processes that ``follow'' the internal transitions of the other processes 
such that always the same guards as in the original run will be satisfied. 
Well-behavedness ensures that this is always possible. 
%\mn{\swen{more detail?}}
\qed
\end{proof}

While we are in general not interested in systems that \emph{only} communicate through internal transitions and negotiations, we can refine this observation based on the states we are interested in, and allow other types of communication.

To this end, define a transition of a process $P$ to be \emph{\free} if it is (i) an internal transition, (ii) a sending transition of either a broadcast (i.e., a 1-sender action) or a $k$-maximal action, or (iii) a receiving transition $\trans{\ls}{a??,\guard}{\ls'}$ of a broadcast with matching sending transition $\trans{\ls}{a!!,\guard}{\ls'}$. Note that the latter includes negotiation transitions. A path from one state to another is \emph{\free} if all transitions on the path are \free. The idea is that free transitions and paths are only restricted by guards (i.e., the \emph{absence} of processes in certain states), but not by the \emph{existence} of other processes in certain states (as, e.g., a $2$-sender transition would be, since a sender depends on the presence of another sender to be able to fire the global transition and move along its own local transition).

\begin{lemma}
\label{lem:cutoffnoreset}
Let $P=\tuple{A,\LS,\ls_0,T}$ be a \wellbehaved GSP process, and $\ls \in \LS$ such that \textbf{all} paths from $\ls_0$ to $\ls$ in $P$ are \free. Then reachability of $s$ is \nice in $P$.
\end{lemma}

\begin{proof}
The argument follows the same line as the one above for protocols with only internal transitions and negotiations, since the same transitions for existing processes are also possible if we can ensure that the same guards can be satisfied in the bigger system. Well-behavedness ensures that there is a run in the bigger system where the same guards are satisfied.
%\mn{\swen{more detail?}}
\qed
\end{proof}

We require that all paths be free, since \emph{existence} of a free path is not sufficient in general: if $m>1$, then the first process that moves along that free path may force other processes to leave it (e.g., by taking a sending transition of a broadcast). However, this condition is still slightly restrictive, and can be relaxed.

%\input{cutoffExample}

%\input{construct}

%\nour{The first observation on this construct, is that it has \emph{infinite} paths from $\ls_0$ (Initial) to $\ls$ (Next1) (because of the cycles). Hence, we need some form of induction to relate such paths in the general case. Otherwise, we can have a specific lemma targeting this exact construct.\\
%The goal here is to come up with a non-recursive definition that can be checked in one pass. Hence, we should avoid marking broadcast receives as ``free'' so we do not introduce any new free paths which would trigger the Lemma again.}

Define a \emph{simple} path as a path with no repeated states. We show that under additional conditions, it is enough to consider restrictions that are based on paths that are simple \emph{and} free:

\begin{lemma}
\label{lem:newcutoffLemma22}
Let $P=\tuple{A,\LS,\ls_0,T}$ be a \wellbehaved GSP process, $\ls \in \LS$, and let $\mathcal{F}$ be the set of simple free paths from $\ls_0$ to $\ls$. If for each send transition: 
\begin{compactenum}
\item the transition does not appear in paths in $\mathcal{F}$ and the corresponding receiving transitions $\trans{\ls_s}{a??,G_a}{\ls_d}$ with $\ls_s \in p$ for some $p \in \mathcal{F}$ have $\ls_d=\ls_s$, or,
\item the transition appears in paths in $\mathcal{F}$ and the following holds for every corresponding receive transition $\trans{\ls_s}{a??,G_a}{\ls_d}$ where $\ls_s \in p$ for some $p \in \mathcal{F}$ and $\ls_d \notin p$ for any $p \in \mathcal{F}$: either (a) there exists an internal transition $\trans{\ls_s}{}{\ls_d'}$ with $\ls_d' \in p$ for some $p \in \mathcal{F}$, or (b) all paths out of $\ls_d$ lead back to a state $\ls_f$ in a path in $\mathcal{F}$ and are free between $\ls_d$ and $\ls_f$. 
%\footnote{Note that, while we say free here, the first item implicitly disallows such paths to have broadcast-send transitions.},
\end{compactenum}
then reachability of $\ls$ is \nice in $P$.
\end{lemma}

\begin{proof}
%This proof is almost identical to the proof of \lemref{newcutoffLemma}; the only difference is that the ``diverging'' receivers (from free paths) are allowed a larger degree of freedom for ``merging back'' to free paths.
First consider a run of a system that satisfies the above conditions, and has $n>m$ processes, where eventually $m$ of them reach $\ls$. We can simulate this run in a system with $m$ processes by keeping the $m$ processes that reach $\ls$, and removing all others. Note that the sending transitions are on the same free simple path from which processes can diverge using the corresponding receiving or sending transitions, or they do not affect them at all. Hence, at least one of the senders is guaranteed to reach $\ls$. All other senders and receivers may diverge from a simple free path but are guaranteed a free path back to a state along a free path and hence, can reach $\ls$ freely. 

Now assume that all processes in a system of size $m$ eventually reach $\ls$, then we can simulate this run in a bigger system by adding processes (that behave in the same way as an existing process). Note that, since any transition diverging from a free simple path can \emph{only} be triggered by a sending transition on that same free path, it is impossible to add a sender that can make processes diverge and then not reach $\ls$ after.
\qed
\end{proof}

\begin{example}
In this example we show how \lemref{newcutoffLemma22} applies to the example in \figref{MappedIntroExample}. Here $\ls_0$ is the \st{\Env} state, $\ls$ is the \st{\Report} state, and the value of $m$ is 3 (since the safety specification is: no more than 2 detectors can report the fire).

The set of simple free paths $\mathcal{F}$ is:
\begin{compactenum}[--]
\item $\st{\Env} \xrightarrow{} \st{\Ask} \xrightarrow{\trr{\Smoke}!!} \st{\Pick} \xrightarrow{\trr{\Choose_i}!!} \st{Report}$ for $i \in \{1,2\}$, and 
\item $\st{\Env} \xrightarrow{} \st{\Ask} \xrightarrow{\trr{\Smoke}??} \st{\Pick} \xrightarrow{\trr{\Choose_i}!!} \st{Report}$ for $i \in \{1,2\}$.
\end{compactenum}
%It is clear that all the sending transitions {$\trr{\Smoke}!!,\trr{\Choose_1}!!,\trr{\Choose_2}!!$} appear only in $\mathcal{F}$. Furthermore, along those paths, we have the following broadcast-receive transitions that need to satisfy the second condition of the lemma:
%\begin{compactenum}[--]
%\item the transition $\trans{\st{\Env}}{\trr{\Smoke}??}{\st{\Idle}}$ satisfies the condition because the internal transition $\trans{\st{\Env}}{}{\st{\Ask}}$ exists in a path in $\mathcal{F}$.
%\item the transition $\trans{\st{\Pick}}{\trr{\Choose}??}{\st{\Idle}}$ satisfies the condition since all paths out of \st{\Idle} are free (namely, the negotiation transition $\trans{\st{\Idle}}{\trr{\Resett}}{\st{\Env}}$) and lead back to a path in $\mathcal{F}$.
%\end{compactenum}

It is clear that all the sending transitions {$\trr{\Smoke}!!,\trr{\Choose_1}!!,\trr{\Choose_2}!!$} appear only in $\mathcal{F}$. Furthermore, the corresponding broadcast-receive transitions satisfy the required conditions as follows:
\begin{compactenum}[--]
\item the transition $\trans{\st{\Env}}{\trr{\Smoke}??}{\st{\Idle}}$ satisfies condition (2a) because the internal transition $\trans{\st{\Env}}{}{\st{\Ask}}$ exists in a path in $\mathcal{F}$.
\item the transition $\trans{\st{\Pick}}{\trr{\Choose}??}{\st{\Idle}}$ satisfies condition (2b) since all paths out of \st{\Idle} are free (namely, the negotiation transition $\trans{\st{\Idle}}{\trr{\Resett}}{\st{\Env}}$) and lead back to a path in $\mathcal{F}$.
\end{compactenum}

Since \lemref{newcutoffLemma22} holds, the reachability of $\ls$ is \nice and the cutoff is 3. 

\end{example}
%
%Note that, while our cutoff theorems only cover a restricted class of 
%systems, they are sufficient for many applications, as we show in 
%Section~\ref{sec:evaluation}.

\paragraph{Checking the Cutoff Conditions.}
Note that while the conditions in \lemref{newcutoffLemma22} seem complex, all our 
cutoff conditions can be checked on the process definition in polynomial 
time, making them well-suited for fully automatic verification.
\section{Applications and Evaluation}
\label{sec:evaluation}

To evaluate our approach, we consider several distributed 
applications that use agreement protocols like consensus or leader election, and that can be modeled as well-behaved systems that satisfy one of our cutoff lemmas:
%\mn{\nour{old: To evaluate our approach, we consider several distributed 
%applications that use agreement protocols like consensus or leader election, and show that each can be modeled as a well-behaved system that satisfies one of our cutoff lemmas:}}

\begin{compactenum}[--]
\item Chubby~\cite{burrows2006chubby}: A distributed lock service for coarse-grained synchronization with an elected leader node that handles client messages.

%\para{Chubby Application.}
%Chubby~\cite{burrows2006chubby} is a distributed lock service. Applications can %interface
%with Chubby as a file system where they send reads and writes to the Chubby
%system, and the data is replicated safely on different servers. The Chubby
%system starts by picking a \emph{leader} server which is responsible for receiving
%clients' requests. The leader can safely serve read requests directly and is
%responsible for committing writes to all the replicas before sending an
%acknowledgment to the client. The leader periodically times out, and a new round
%of election should happen to pick a new leader. 

\item Distributed Smoke Detector (SD): A sensor network
application that elects a subset of processes, who have detected smoke, to report
to the authorities.

%\para{Smoke Detector (SD).} We consider a distributed smoke detector whose behavior is as follows. Upon detecting smoke, the
%detectors coordinate using an \consagree protocol to choose at most two processes
%to contact the fire department. This example displays the utility of allowing the set of participating processes in the \consagree operation to be dynamically constructed: only detectors that detected the smoke coordinate as opposed to all the detectors in the system.

\item Smoke Detector with Reset (SDR): A variant of SD that
uses a ``reset'' signal to resume monitoring for smoke, thereby requiring
infinite rounds of \consagree.
(this was our motivating example in \figref{MappedIntroExample})

%\para{Smoke Detector with Reset (SDR).} We examine a variant of the smoke detector that uses a ``reset'' environment signal to move all detectors back to the initial state to start new rounds of detection. The presence of the reset signal makes the problem more interesting as it simulates infinite rounds of \consagree.

\item Distributed Mobile Robotics (DMR): Based on an existing
benchmark~\cite{Desai.DRONAFrameworkSafe.X.2017}, where a set of robots
successively coordinate to create a motion plan.

%\para{Distributed Mobile Robotics (DMRs).} As a larger case study, we model the system presented in~\cite{Desai.DRONAFrameworkSafe.X.2017} where a set of robots share a workspace with obstacles, and need to coordinate their movements. The robots coordinate to create a motion plan by successively choosing each robot to create a plan while taking into account the previous robots' plans. We model this system in our \ch model: the robots choose one robot to make a plan, then the remaining robots re-enter consensus to choose a second robot and so on. 
%As in SATS, each round of consensus determines the participant set of the next, but unlike SATS, the number of consensus rounds is not statically bounded.

\item Distributed Key-Value Store (KVS) modeling a key-value
store \'a la Redis~\cite{radisref}.

%\para{Distributed Key-Value Store (KVS)} For a proof-of-concept of the \valstore operation, we model a simple key-value store in the spirit of Redis~\cite{radisref}\nour{let's find something distrusted, as opposed to in-memory?}, where read and update requests come from the clients to be served by processes which reach \consagree on their stored value. We forgo the logic of key indexing in favor of a more readable, single-value system. Processes service read requests from their stored value, while upon receiving an update request, they use \consagree to determine the new consistently stored value. This system shows the power of our abstraction: we allow the user to start \consagree from multiple locations in the process.
 
\item Small Aircraft Transportation System (SATS): The landing
protocol of SATS proposed by NASA~\cite{satsref}. SATS aims to increase access to small airports without control towers by allowing aircrafts to coordinate with each other to operate safely upon entering the airport airspace.

%\para{Small Aircraft Transportation System (SATS).}
%Our final case study addresses the landing operation of the Small Aircraft
%Transportation System (SATS)~\cite{satsref}, a frequent target for verification~\cite{johnson2012parametrized,carreno2005safety,munoz2004modeling}.
%The idea is that the aircraft should coordinate with each other to figure out
%how to land safely and avoid collisions. The planes use consensus to choose successive subsets of planes to progress to the next phase of landing, until just one plane at a time is chosen to land. 
%SATS uses multiple rounds of consensus, with the results of one consensus feeding in to the participant set of the next.

\item SATS$^{++}$: A variant of the SATS
protocol where all processes communicate explicitly to determine subsets of aircrafts to coordinate the landing with.\\

%\para{Variant of SATS (SATS$^{++}$).} \nour{important: we can keep this only if we keep the one where we use pairwise!}We consider a variant of the SATS protocol where all processes communicate to build \code{BaseSet}, rather than only the processes in \st{Ask} 
%\end{compactenum}
\end{compactenum}
\begin{table}[t] 
\begin{tabular}{l  c  c  c  c}
\hline 
\toprule 
\multirow{2}{*}{\textbf{Benchmark}} & \multirow{2}{*}{\textbf{ States }} &\multirow{2}{*}{\textbf{ Cutoff }}  &\textbf{ Verification} \\
 							       & 							      & 							      & \textbf{Time(s)}\\
 \midrule 
Chubby 				& 9 		& 2 		& 0.12 \\ 
SD 					& 5 		& 3  	& 0.28 \\ 
SDR 				& 5 		& 3 		& 0.13 \\ 
DMR 				& 8 		& 3  	& 0.16 \\ 
KVS 				& 18 	& 3  	& 3.06 \\ 
SATS 				& 24 	& 5  	& 3.83 \\ 
SATS$^{++}$ 		& 26 	& 5  	& 17.1 \\ 
\bottomrule 
\end{tabular}
%
%\nour{saving a line?\\}
%\begin{tabular}{l  c  c  c  c}
%\hline 
%\toprule 
%\textbf{Benchmark  } & \textbf{  States  } &\textbf{  Cutoff  }  &\textbf{  Verification Time(s)} \\
% \midrule 
%Chubby 				& 9 		& 2 		& 0.12 \\ 
%SD 					& 5 		& 3  	& 0.28 \\ 
%SDR 				& 5 		& 3 		& 0.13 \\ 
%DMR 				& 8 		& 3  	& 0.16 \\ 
%KVS 				& 18 	& 3  	& 3.06 \\ 
%SATS 				& 24 	& 5  	& 3.83 \\ 
%SATS$^{++}$ 		& 26 	& 5  	& 17.1 \\ 
%\bottomrule 
%\end{tabular}

%\vspace{0.5em}
\caption{Performance of parameterized verification based on our cutoffs.}%\nour{merge local vars and locs into states? not sure.}}
\label{table:results}
\end{table}
%
%\mnnn{Not sure how to go about this. We need a reason to do this in the choose model and not here. And do we want to say more about the choose model?}
\begin{wrapfigure}[10]{r}{0.45\textwidth}
\vspace{-2em}
%\vspace{-2em}
\centering
% trim=left botm right top.
%\fbox{\includegraphics[width=1\linewidth,clip, trim=2cm 0cm 2cm 0cm]{../common/varyNlog.pdf}}
\includegraphics[width=1\linewidth,clip, trim=2cm 0cm 2cm 0.3cm]{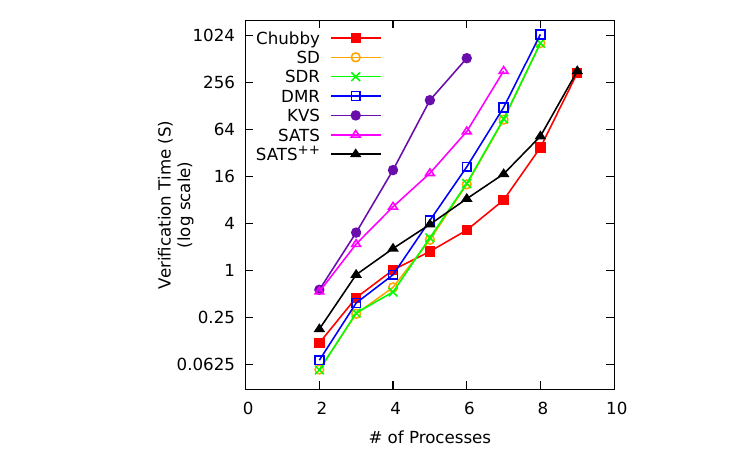}
\caption{Verification time as a function of the number of processes.}
\label{fig:varyN}
\vspace{-1em}
\end{wrapfigure}

In addition, we provide an experimental evaluation, based on related
work~\cite{choosepaper.arxiv} in which a new model---the 
  \ch model---that can be seen as a
  refinement of GSP, is proposed. The \ch model extends a standard model of
  distributed
  systems~\cite{Alur.AutomaticCompletionDistributed.X.2015,Alur.AutomaticSynthesisDistributed.X.2017}
  with a primitive that abstracts various types of distributed \consagree
  protocols. The work further defines a mapping from the \ch model 
			 to GSP that
  establishes a simulation equivalence between the two models, enabling
  interchange of safety verification and cutoff results between the two
   models. 
		
	To make use of the ease of encoding the above benchmarks in the \ch model and
  the ease of verification in the \ch model using off-the-shelf model
  checkers, we illustrate the effect of our cutoff results on 
  efficiency of verification in the \ch model.  For the benchmarks given above,
  \figref{varyN} depicts the verification time as a function of the number of
  processes. Observe that verification time grows roughly exponentially with
  the number of processes. Moreover, verification for all the benchmarks timed
  out beyond 9 processes, for a timeout of 30 minutes. In contrast, in
  \tabref{results} all benchmarks have a cutoff of less than 6, and reasonable verification times.

\section{Related Work}

Bodies of work that aim at automatically solving the parameterized verification problem
(which is undecidable in the most general
case~\cite{Suzuki.PMCP_UndecidableFirstPPR.1988.PPR,Emerson.PMCP_UndecidableNeatPaper.POPL.1995})
take a large variety of different
approaches~\cite{GurfinkelSM16,BouajjaniJNT00,PnueliRZ01,ClarkeTV06,GhilardiR10,abdulla2016parameterized,KaiserKW10,KurshanM95,WolperL89},
in most cases without a focus on decidability. In the following we consider the
approaches that target decidability, with models closely related to our GSP model.

\paragraph{Models with Broadcasts and/or Global Guards.} We want to enable reasoning about distributed systems, abstracting complex building blocks like agreement protocols by primitives that satisfy \emph{assume-guarantee} specifications. To support parameterized reasoning for systems with such abstractions, one needs a model with (i) conjunctive guards to model the assumptions, and (ii) forms of synchronization that are sufficiently general to model the guarantees of those building blocks, i.e., generalizations of broadcast communication. 

Esparza et al.\cite{EsparzaFM99} present a decidability result for safety properties of broadcast protocols, but without global guards. Their result is also based on a reduction to WSTSs, but we showed that the WQO presented in their work (corresponding to the WQO $\wqo$ in \secref{decUGSPs}) is not suitable for systems with guards. We note that our \hlm subsumes the model of Esparza et al., and that our cutoff results also apply to their model (which had no previous cutoff results).

Other existing models either are not sufficiently general~\cite{EK00,EK03a,EK03b}, or support a combination of broadcasts and conjunctive guards without restrictions~\cite{emerson2003model}, which makes safety undecidable. This highlights the significance of our result: we manage to find a model with conjunctive guards and global synchronization such that safety remains decidable.

\paragraph{Other Decidable Classes.} One way to obtain decidability is to restrict the generality of the parameterized verification problem in various ways.
% However, by restricting some aspect of the system, decidability results can be obtained. One large focus of research in this area was the identification of interesting classes of systems, and proving the decidability or undecidability of the resulting verification problems. 
 Most results in this direction consider a fully connected network (a clique), either with rendezvous communication~\cite{GS92,AminofKRSV18}, local updates with global guards~\cite{EK00,AusserlechnerJK16}, or variants of these~\cite{DelzannoRB02}. Some communication primitives have also been considered in more complex networks, for example token passing~\cite{EN03,CTTV04,AminofJKR14}, or broadcasts~\cite{DelzannoSTZ12}. Decidability results for systems that are composed of identical components have recently been surveyed by Bloem et al. \cite{BloemETAL15} as well as Espazra et al. \cite{Esparza16}.
%\mn{a comparison note here will be the same as the one under cutoffs.}
%\para{Cutoff results.} 
Several bodies of work attempt to identify cutoff bounds for different classes of distributed systems. For example, cutoffs have been obtained for cache coherence protocols~\cite{EK03a}, guarded protocols~\cite{Jacobs.AnalyzingGuardedProtocols.X.2018,emerson2003model,EK00}, consensus protocols~\cite{Maric.CutoffBoundsConsensus.X.2017}, and self-stabilizing systems~\cite{Bloem.SynthesisSelfstabilisingByzantineresilient.X.2016}. 
None of these approaches are sufficiently general to tackle the types of distributed applications we address. 

%\nour{not sure this fits under decidable classes (where we had it), or that it's super related, maybe we take it out?:} Approaches based on threshold 
%automata~\cite{Konnov.CompletenessBoundedModel.X.2014,Konnov.Para2ParameterizedPath.X.2017,Konnov.ShortCounterexampleProperty.X.2017} build on the intuition that most fault-tolerance strategies depend on some form of counting replies and checking if they are above some threshold.

\paragraph{Petri Nets and Vector Addition Systems.} 
Also closely related to the parameterized verification problems we consider is the body of work on Petri nets and vector addition systems, surveyed e.g. by Esparza and Nielsen~\cite{EsparzaN94} or Reisig~\cite{Reisig13}. While some types of communication can faithfully be expressed in these systems, global synchronization in general cannot.

\section{Conclusion}

We introduced global synchronization protocols (GSP), a system model that 
generalizes many existing models supporting global 
synchronization such as broadcast synchronization, 
pairwise rendezvous, and asynchronous rendezvous. 
We identified sufficient conditions, 
summarized under our notion of well-behavedness, that ensure 
decidability of the parameterized verification problem even in the presence 
of global (conjunctive) transition guards. Finally, we investigated cutoffs 
for parameterized verification, and identified sufficient conditions under which small 
cutoffs exist.

In ongoing work, we are focusing on extensions of our cutoff results
as well as a dedicated implementation of our decision procedure.
In the near future, we plan to investigate sufficient conditions that enable 
support for the parameterized verification of liveness properties for GSPs,
and intend to develop a domain-specific 
language for writing GSPs that are well-behaved by construction.

\bibliographystyle{splncs04}
\bibliography{ms}

\newpage{}
% Appendix
\appendix
\section{\hlm Details}\label{app:pvapp}

\subsection{The Environment Process in \gbc Protocols}
\label{app:gbc-env}
While we do not consider an environment process in the \hlm, all of our results extend to the case with a distinguished process (or even multiple processes) that is not replicated. 
 In fact, in the \hlm a system of the form $P_1 \ \| \ \ldots\ \| \ P_{n} \ \| \ E$ can be simulated by a system of the form $P_1' \ \| \ \ldots\ \| \ P_{n}' \ \| \ P_{n+1}'$ by letting $P'$ be a process that starts in an initial state where only a broadcast (on a fresh broadcast action) is possible, where the sender moves to the initial state of $E$ and all receivers move to the initial state of $P$. Moreover, the environment process is trivially well-behaved under the assumption that guards of regular processes do not depend on the state of the environment process, and that transitions of the environment process are not guarded.

\end{document}